\PassOptionsToPackage{table}{xcolor}
\documentclass[10pt, conference]{IEEEtran}%\usepackage{colortbl}
\usepackage[a4paper, left=1.2cm, right=1.2cm, top=1.8cm,bottom=2.5cm]{geometry}
\usepackage{rotating}
\usepackage{algpseudocode}
\usepackage{array}
\usepackage{graphicx}
\usepackage{algorithm}
\usepackage{times}
\usepackage{comment}
\usepackage{algorithm,url}
\usepackage{epsfig}
\usepackage{xspace}
\usepackage{bbm,bm}
\usepackage{amsmath}
\usepackage{amsthm}
\usepackage{amsfonts}
\usepackage{breqn}
\usepackage[small]{caption}
\usepackage[english]{babel} % English and German language 
\usepackage{booktabs} % Horizontal rules in tables 
\usepackage{comment} % Necessary to comment several paragraphs at once
\usepackage[utf8]{inputenc} % Required for international characters
\usepackage[T1]{fontenc} % Required for output font encoding for international characters
\usepackage{multirow}
\usepackage{amsmath,amsfonts,amsthm} % Math packages which might be useful for equations
\usepackage{tikz} % For tikz figures (to draw arrow diagrams, see a guide how to use them)
\usepackage{tikz-cd}
\usetikzlibrary{positioning,arrows} % Adding libraries for arrows
\usetikzlibrary{decorations.pathreplacing} % Adding libraries for decorations and paths
\usepackage{tikzsymbols} % For amazing symbols ;) https://mirror.hmc.edu/ctan/graphics/pgf/contrib/tikzsymbols/tikzsymbols.pdf 
\usepackage{blindtext} % To add some blind text in your paper
\usepackage{tikz}
\usepackage{subcaption}
\usepackage[noabbrev,capitalize]{cleveref}
\usepackage{enumitem}
\usepackage{pifont}

\renewcommand{\paragraph}[1]{\vspace*{0.05in}\noindent\textbf{#1}}

\crefformat{section}{Sec. #2#1#3}

\usetikzlibrary{fit,tikzmark}

\newcommand{\TBD}{\textcolor{red}{TBD}\xspace}	

\newtheorem{lemma}{Lemma}

\newtheorem{theorem}{Theorem}

\newtheorem{proposition}{Proposition}

\newcommand{\D}{\mathcal D}
\renewcommand{\S}{\mathcal S}										% state space 

\newcommand{\A}{\mathcal A}											% action space 
\newcommand {\p}{\alpha}									         	
\newcommand {\lbda}{\zeta}							
\newcommand {\Lbda}{Z}

\newcommand{\E}[2]{\mathop{\mathbb{E}}_{#1}\!\left [\xspace#2\xspace\right ]\xspace}		% expectation
\newcommand{\Pro}[1]{\mathbb{P} \rp{#1} }		% probability 
\newcommand {\R}{\mathbb{R}}								        		% real numbers

\newcommand {\M}{\mathcal{M}}		
\newcommand{\rp}[1]{\left( #1 \right)}
\newcommand{\expval}[1]{\mathbb{E}\left[ #1 \right]}
\newcommand{\expvalDist}[2]{\mathbb{E}_{#1} \left[ #2 \right]}
\newcommand{\secondProblem}{load balancing\xspace}
\IEEEoverridecommandlockouts 
\newcommand{\extver}[1]{\cref{#1}}

\begin{document}
%%%%%%%%%%%%%%%%%%%%%%%%%%%%%%%%%%%%%%%%%%%%%%%%%%%%%%%
%%%%%%%%%%%%%%%%%%%%%%%%%%%%%%

\title{Optimal Flow Admission Control in Edge Computing via Safe Reinforcement Learning}

%%%%%%%%%%%%%%%%%%%%%%%%%%%%%%
%%%%%%%%%%%%%%%%%%%%%%%%%%%%%%%%%%%%%%%%%%%%%%%%%%%%%%%%%%%%%%%%%

\author{A. Fox$^{\diamond}$, F. De Pellegrini$^{\diamond}$, F. Faticanti$^{\star}$, E. Altman$^{\dagger}$, and F. Bronzino $^{\star}$\thanks{$^{\diamond}$ LIA, Avignon university, Avignon, France; $^{\star}$ENS, Lyon, France; $^{\dagger}$INRIA, Sophia Antipolis, France.}}
\maketitle

\begin{abstract}
With the uptake of intelligent data-driven applications, edge computing
infrastructures necessitate a new generation of admission control algorithms to
maximize system performance under limited and highly heterogeneous resources. In
this paper, we study how to optimally select information flows which belong to
different classes and dispatch them to multiple edge servers where %deployed
applications perform flow analytic tasks. The optimal policy is obtained via the theory of 
constrained Markov decision processes (CMDP) to take into account the demand of
each edge application for specific classes of flows, the constraints on 
computing capacity of edge servers and the constraints on access network capacity. 

We develop DRCPO, a specialized primal-dual Safe Reinforcement Learning (SRL)
method which solves the resulting optimal admission control problem by reward
decomposition. DRCPO operates optimal decentralized control and mitigates
effectively state-space explosion while preserving optimality. Compared to
existing Deep Reinforcement Learning (DRL) solutions, extensive results show
that it achieves $15$\% higher reward on a wide variety of environments, while
requiring on average only $50$\% learning episodes to converge.
Finally, we further improve the system performance by matching DRCPO with load-balancing 
in order to dispatch optimally information flows to the available edge servers.
\end{abstract}

\pagestyle{plain}
%%%%%%%%%%%%%%%%%%%%%%%%%%%%%%%%%%%%%%%%%%%%%%%%%%%%%%%%%%%%%%%%%%%%%%%%%%%%%%%%%%%%%%%%%%%%%%%%%%%%%%%%%%%%%%%%%%%%%%%%%
%%%%%%%%%%%%%%%%%%%%%%%%%%%%%%%%%%%%%%%%

\section{Introduction}\label{sec:intro} 

%%%%%%%%%%%%%%%%%%%%%%%%%%%%%%%%%%%%%%%%
%%%%%%%%%%%%%%%%%%%%%%%%%%%%%%%%%%%%%%%%%%%%%%%%%%%%%%%%%%%%%%%%%%%%%%%%%%%%%%%%%%%%%%%%%%%%%%%%%%%%%%%%%%%%%%%%%%%%%%%%%
Edge computing techniques have emerged in recent years as a powerful solution to locally process a variety of information flows. Facing the need of serving exponentially growing service demands, infrastructure and service providers have responded by deploying  their resources, from processing to storage, at the network edge. Processing information as close as possible to its source significantly reduces the amount of data to transfer to remote cloud locations, thus decreasing latency and overhead during remote service access~\cite{LuoEdgeSurvey2021,hu2023edge}. Enabled by edge clouds, new classes of data intensive AI-based applications~\cite{hu2023edge,wang2018bandwidth,pakha2018reinventing,seufert2024marina,borgioli2023real} are now widespread. Unfortunately, while edge clouds offer an on premise computing solution, they are easily overwhelmed when demand exceeds available resources.

In fact, in contrast to the previous data center driven cloud model, edge clouds are often co-located with the existing network equipment and deploy limited computational resources. Thus, they can host a limited number of applications at any point in time. This generates the need of carefully designing solutions to orchestrate  the operations of deployed applications. For instance, existing edge-based  solutions often aim to efficiently configure available computing resources~\cite{hung2018videoedge,HuangDynAC2022,jiang2018chameleon,zhang2019hetero}%zhang2019hetero,zhang2017live}
or attempt to manipulate how data flows are transported to reduce the transmission overhead~\cite{pakha2018reinventing}. This is indeed a major concern especially in smart-city environments \cite{Khan2019ECSmartCities}. Yet, as the number of applications 
and, more significantly, the number of information flows increase, the need for a new generation of admission control algorithms becomes apparent.  

Admission control is essential for managing resources efficiently, preventing under-utilization and degradation of service quality. It is widely used across various communication and computing systems, including mobile networks \cite{senouci2004call,raeis2020reinforcement}, web services \cite{Cherkasova2001}, optical networks \cite{Sue2011}, and cloud computing \cite{Konstanteli2014,Sajal_OSDI2023}. However, the performance of AI-based edge applications depends not just on networking or compute metrics but also on the information content, posing new challenges for admission control algorithms. When deployed at the edge, admission control algorithms must select information flows processed on edge servers to maximize the information extracted by deployed applications. Flow arrivals and departures affect application operations, especially when information flow sources are mobile nodes entering or leaving an area. Edge service virtualization allows replicating multiple instances of applications and deploying them on several servers simultaneously. Replication enhances robustness but requires precise performance considerations. The results obtained in this work highlight the need to orchestrate flow admission by considering the actual installation of compute modules on edge servers and the required access bandwidth.

Modern edge applications can be commonly characterized by five features: applications process flows generated by a large number of sources of different nature; these flows can enter or leave the architecture over time due to various events; the edge infrastructure deploys a set of applications to process the flows on edge servers which are equipped with a given amount of resources (e.g., compute and memory); finally, the distributed nature of both sources and edge servers imposes the implementation of a control plane mapping flows to compute infrastructure.

Earlier models for admission control in edge-computing systems have not yet addressed all of these challenges. Hence, in this paper we develop new theoretical foundations for the edge admission problem. We extend models originally developed for admission control in loss systems, which established the paradigmatic concept of trunk-reservation \cite{miller1969queueing}. In those early models, a finite service pool is made available to a finite set of service classes and each class is associated a certain reward for the admission of one of its customers. Markovian single-queue models for trunk-reservation have been studied in depth \cite{feinberg1994,feinberg2006,miller1969queueing,IGA}. %On the other hand, while some multi-server admission control techniques have been studied for cloud computing, the accent there is on virtual machine placement with respect to pricing 
While some multi-server admission control techniques have been studied for cloud computing, the focus is primarily on virtual machine placement relative to pricing \cite{Konstanteli2014} or overbooking \cite{Sajal_OSDI2023}. Once applications are placed onto edge servers, the framework considered in this work provides an optimal decentralised flow admission control logic. This necessitates several novel contributions: 

{\noindent\em System model (\cref{sec:sysmod}).} We develop a novel constrained Markov decision model to capture 
the dynamic admission control and load balancing of information flows originating from multiple sources. It accounts for heterogeneous capacity constraints for both access network and edge servers. It also includes applications' replication on multiple servers and their preferences on the classes of information flows they process. \\
{\noindent\em Solution concept (\cref{sec:cmdp}).} Using constrained Markov decision theory, we have derived the structural properties of the optimal decentralized admission control policy, showing it requires at most one randomized action per server. The result is not obvious since servers' states are reward-coupled. \\
{\noindent\em A new learning algorithm (\cref{sec: learning optimal admission policy}).} We introduce new tools to optimize mobile information admission control policy rooted in SRL. DRCPO is a novel actor-critic scheme that leverages the structure of the optimal solution to implement the optimal flow admission policy effectively. It is tailored for cases where the same application may be installed on several edge servers simultaneously. \\
{\noindent\em Load balancing (\cref{sec: optimizing routing policy}).} Finally, a two-stages joint optimization procedure increases further the system performance by jointly optimizing routing and admission control.

Our numerical results (\cref{sec: numerical results}) demonstrate that, by leveraging  the properties of the underlying Markovian model, not only it is possible to learn the optimal admission policy with no approximation, but this can be attained with a significant reduction in complexity with respect to state of the art techniques, which are typically oblivious to the structure of the optimal policy and value function. %More specifically, leveraging results from reinforcement learning with reward decomposition, DRCPO achieves convergence while requiring only 50\% of the episodes compared to existing baselines. Importantly, DRCPO is provably optimal, unlike competing baselines, and it outperforms popular function approximation techniques employing deep neural networks (e.g., RCPO~\cite{tessler2018reward}). By iterating an open-loop stochastic approximation algorithm for \secondProblem, followed by a policy improvement step for admission control, we attain a significant increase of the overall system performance. Finally, our tests demonstrate that, in scenarios subject to heterogeneity in server capacity, applications' utility, and access bandwidth, when load balancing is jointly  optimized with the servers' admission control, the system performance greatly benefits from replicating applications on several edge servers.

% The rest of the paper is organized as follows. 
% In \cref{sec:use cases}, we outline practical use-cases of the introduced admission control model. \cref{sec:sysmod} presents the Markovian model for flow edge computing. The constrained admission control problem is detailed in \cref{sec:cmdp}, while \cref{sec: learning optimal admission policy} presents efficient reinforcement learning algorithms to determine an optimal solution. \cref{sec: optimizing routing policy} addresses routing optimization to address the second problem. The main numerical results are outlined in \cref{sec: numerical results}. Related works are discussed in \cref{sec:related}, while \cref{sec:conclusions} concludes the paper.

%%%%%%%%%%%%%%%%%%%%%%%%%%%%%%%%%%%%%%%%%%%%%%%%%%%%%%%%%%%%%%%%%%%
%%%%%%%%%%%%%%%%%%%%%%%%%%%%%%%%%%%%%%%%

%%%%%%%%%%%%%%%%%%%%%%%%%%%%%%%%%%%%%%%%%%%%%%%%%%%%%%%%%%%%%%%%%%%
\begin{figure}[t!]
     \centering
     \begin{subfigure}[c]{0.40\columnwidth}
         \centering
         \includegraphics[width=\textwidth]{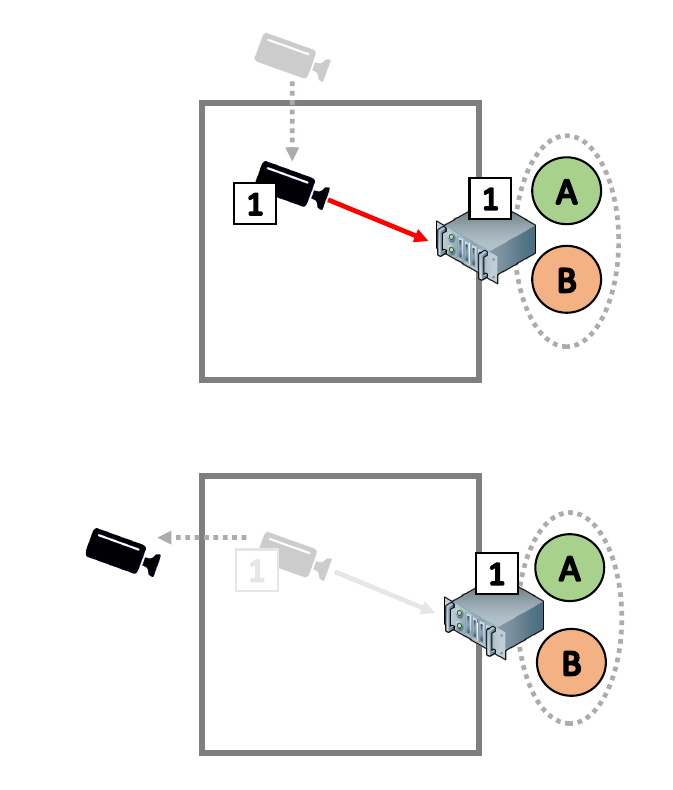}
         \caption{System events}\label{fig:arrival_and_departure}
     \end{subfigure}
     \hfill
     \begin{subfigure}[c]{0.58\columnwidth}
         \centering
         \includegraphics[width=\textwidth]{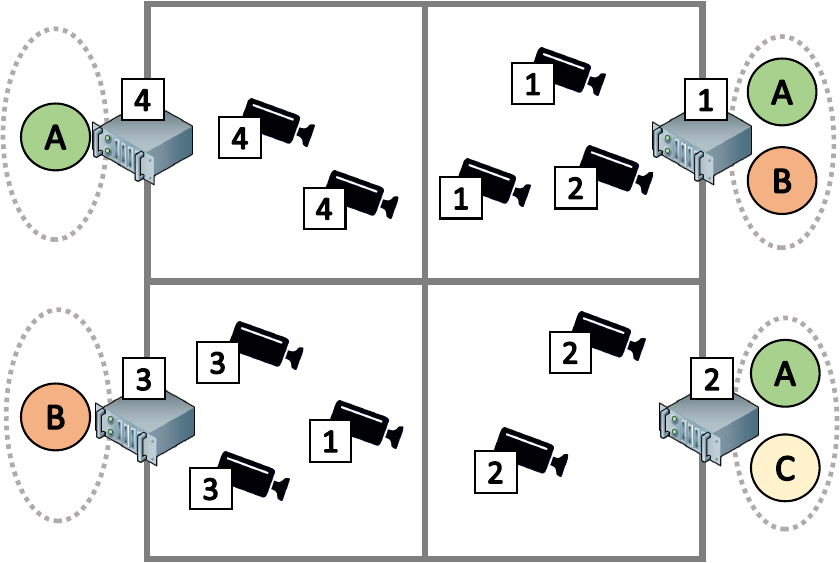}
         \caption{System state}\label{fig:state}
     \end{subfigure}
        \caption{(a) {\em Camera arrival and departure:} a camera arrives in area, it transmits its flow towards a tagged server (boxed index), then departs; (b) {\em System state:} using notation in Tab. \ref{tab:notation}, $M = 4$; $\D^1 = \{ A, B \}; \D^2 = \{ A, C \}; \D^3 = \{ B \}; \D^4 = \{ A \}$; $X^1 = (2, 1, 0, 0); X^2 = (0, 2, 0, 0), X^3 = (1, 0, 2, 0); X^4 = (0, 0, 0, 2)$; $j = 1, i= 1$; $Y^1 = 3, Y^2 = 2, Y^3 = 3, Y^4 = 2$.}
        \label{fig:system}
        %\vspace{-4mm}
\end{figure}

%%%%%%%%%%%%%%%%%%%%%%%%%%%%%%%%%%%%%%%%%%%%%%%%%%%%%%%%%%%%%%%%%%%%%%%%%%%%%%%%%%%%%%%%%%%%%%%%%%%%%%%%%%%%%%%%%%%%%%%%%%%%%%%%
%%%%%%%%%%%%%%%%%%%%%%%%%%%%%%%%%%%%%%%%%%%%%%%

\section{System Model}\label{sec:sysmod}

%%%%%%%%%%%%%%%%%%%%%%%%%%%%%%%%%%%%%%%%%%%%%%%
%%%%%%%%%%%%%%%%%%%%%%%%%%%%%%%%%%%%%%%%%%%%%%%%%%%%%%%%%%%%%%%%%%%%%%%%%%%%%%%%%%%%%%%%%%%%%%%%%%%%%%%%%%%%%%%%%%%%%%%%%%%%%%%%

We introduce a semi-Markov model general enough to cover the main characteristics  of the edge flow admission control problem just outlined. It features the point process of arrivals and departures of flows belonging to a certain class, the coverage requirements of applications installed on edge servers (described by their utility function), the routing of flows to different servers and, finally, a policy to admit flows to edge servers. We now precise its mathematical definition.

Flows belong to class index $j \in\{ 1,\ldots, M\}$. They are generated according to a Poisson process of intensity $\lbda_j$. A flow of class $j$ remains active for an exponential time of mean $1/\mu_j$ seconds, after which it leaves the system. The flow arrival processes of different classes are independent and independent of the servers' occupancy. Edge applications consist of different modules installed on some designated servers; we say that an application is installed on a server if that application has a module deployed there. In \cref{fig:system}b, we have represented a use case for video analytics. There, here each class corresponds to video stream sources situated in one of $M=4$ areas. Each area hosts a designated edge server. An application is installed on multiple servers, for instance application $A$ in Fig.~\ref{fig:system} resides on server $1$, $2$ and $4$. We consider the case of $M$ edge servers: the general case is a straightforward extension.

%For the moment we assume a static load balancing policy. 
Let $u_j^i$ denote the probability that a flow of class $j$ is routed towards server $i$.\footnote{Throughout the paper we use subscript indexes to denote the class of the flow and superscript indexes to denote destination servers.} The aggregated arrival rate at server $i$ is $\Lbda^i=\sum_j u_j^i \lbda_j $, and the total arrival rate $\Lbda=\sum_j \zeta_j$. Let denote $\p_j^i= u_j^i \lbda_j/\Lbda^i$ the probability that an arrival is of class $j$ and it is routed to server $i$. Once routed to server $i$, a flow is either accepted or rejected for service depending on the system state. If accepted at server $i$, it can feed the modules of applications installed on that server.
We further assume perfect information, i.e., different servers are aware of the state of other servers. The decision-making process regarding accepting or rejecting an incoming flow also depends on the number of flows from the same class already processed by the same application across the entire system. The computational capacity of server $i$ allows it to process at most $\psi^i$ concurrent flows simultaneously.

%%%%%%%%%%%%%%%%%%%%%%%%%%%%%%%%%%%%%%%%%%%%%%%%%%%%%%%%%%%%%%%%%%%
\begin{table}[t!]
	\centering
	\begin{tabular}{|p{0.16\columnwidth}|p{0.68\columnwidth}|}
		\hline
		\rowcolor{gray!30}{\bf Symbol} & {\bf Meaning}                         \\
		\hline
		    $M$          & number of classes                       \\
            $\lbda_j$    & arrival rate of class $j$ flows       \\
            $\mu_j$      & mean duration of class $j$ flows       \\
            $u_j^i$      & prob. of routing flows of class $j$ to server $i$        \\
		$\D^i$       & set of applications installed on server $i$;  \\
        $d^i$ & number of applications installed on server $i$\\
            $\phi_d$     & servers on which $d \in \D$ is installed\\
		$\S$         & state space                           \\
            $S$          & state $S=(X,J,I)$, $X=(X^1,\ldots,X^M)$ \\
            $X^i=(X_j^i)$    & flows of class $j$ active on server $i$\\
            $Y^i$        & total occupation $Y^i=\sum_{j=1}^M X_j^i$ of server $i$    \\
             $\A=\{0,1\}$ & action space                           \\
		$\psi^i$        & computational capacity of server $i$  \\
		$\theta^i$   & access capacity of server $i$         \\
		$\chi_j(d)$   & coverage requirement of app. $d$ for class $j$\\
		\hline
	\end{tabular}\caption{Main model notation.}\label{tab:notation}
\end{table}
%%%%%%%%%%%%%%%%%%%%%%%%%%%%%%%%%%%%%%%%%%%%%%%%%%%%%%%%%%%%%%%%%%%
Our semi-Markov decision process extends the models presented in \cite{feinberg1994,feinberg2006}. The continuous process is sampled at each arrival time ${t}$ of an information flow. This results into the discrete-time MDP $\M = (\S ,\A, P)$ \cite{lippman1974}. We define $\S$ the system state, $\A$ the action set and $P$ the MDP probability kernel. Whenever possible, {\em uppercase} notation, e.g., $S$, refers to a random process, and {\em lowercase} notation, e.g., $s$, to its realization. $\D^i$ is the set of applications installed on server $i$. %, e.g., in \cref{fig:system} $\D^1 = \{ A, B \}$. 
Variable $\chi_j(d)\in \{0,1\}^M$ indicates whether application $d$ is interested in flows of class $j$ ($\chi_j(d)=1$) or not ($\chi_j(d)=0$).

\paragraph{System state.} The state is a triple $S(t)=(X(t),J(t),I(t))$: 

\noindent i. $X(t)$ is the matrix representing the system occupation at time $t$, where $X_j^{i}(t)$ denotes the number of flows of class $j$ being routed to server $i$ at time $t$ and being processed by applications in $\D^i$. $Y^i(t):=\sum_j X_j^i(t)$ is the corresponding server $i$ total occupancy.

\noindent ii.  $J(t)$ represents the class of the incoming flow;

\noindent iii.  $I(t)$ is the destination server for the incoming flow.

\paragraph{Action set.} The admission of an incoming flow for processing at a certain server is represented by action $A(S(t)) \in \A \rp{S(t)} \subset \{0,1\}$. Here $A(S(t))=0$ signifies \emph{reject} and $A(S(t))=1$ denotes {\em accept}. If $i$ is the destination server and $Y^i(t)=\psi^i$, then $\A(S(t))=\{0\}$, as server $i$ has no available capacity to host additional flows.  

\paragraph{Probability kernel.} Policy $\pi: \S \rightarrow \A$ associates to state $S(t)$ a probability distribution over action set $\A(S(t))$. Let $p( s' |s,a)$ $=\Pro{S(t+1)= s'|S(t)=s,A(t)=a}$ denote the transition probabilities 
\begin{align}
    p( s'|s,a) &= p((x',j',i')|(x,j,i),a) \notag \\
    &= \p_{j'}^{i'}\;  p({x'}^i_j | x^i_j + a) \mathop{\prod_{k, m = 1}^M}_{(k, m)\not= (i, j)} p({x'}^k_m | x^k_m)
    \label{eq:transition probabilities}
\end{align}
where  $p({x'}^k_m| x^k_m )=\Pro{X^k_m(t+1)=  {x'}^k_m|X^k_m(t)= x^k_m }$.\\\\
Let denote $\widehat p(u; x_j^i)$ the probability of the event that $u$ flows of class $j$ being routed to server $i$ leave in between two arrivals, given that $x_j^i$ flows are active on server $i$: it holds
\begin{equation}\label{eq:deathprob}
	\widehat p(u; x_j^i)= \Lbda \int_0^\infty {x_j^i \choose u} e^{-\mu_j t(x_j^i - u)} (1-e^{-\mu_j t})^u e^{-\Lbda t} dt \nonumber
\end{equation}
for $0 \leq u \leq x_j^i$ and it is zero otherwise. Hence, %we can leverage the independence of the flows to derive component-wise 
the state transition probabilities at server $i$ are derived as 
\begin{equation}\label{eq:transprob}
	p({x'}_j^i | x_j^i,  a)= \widehat p (x_j^i - {x'}_j^i + a^i_j; x_j^i)
\end{equation}
with $a_j^i = a$ if and only if $I(t)=i, J(t) = j$ and $a_j^i = 0$ otherwise. Clearly this probability is nonzero only when ${x'}_j^i \leq x_j^i + a_j^i $.

\paragraph{Rewards.} Let $R_{t+1}$ be the reward attained after the action at time $t$, following the traditional notation in \cite{suttonRL}. In particular, $r(s,a)=\E{}{R_{t+1}|S(t)=s,A(t)=a}$. By admitting a flow of class $j$ to server $i$, the instantaneous reward for applications binding to the tagged flow is expressed as
\begin{equation}
r(s,a) = a \cdot \sum\limits_{d \in \D_i} r_d(x) \chi_j(d)
    \label{eq:immediate reward}
\end{equation}
where $r_d(\cdot)$ is the marginal gain attained by binding a new flow to application $d$.
Later, we define $w_{j, d}$ as the total amount of flows of class $j$ currently being processed by application $d$ in the system, and $\phi_d$ as the set of servers on which application $d$ has been installed. Specifically, $w_{j, d} = \sum_{i \in \phi_d} x^i_j$. The immediate reward considered will only depend on this quantity: $r_d(x) = r_d(w_{j, d})$. Additionally, we assume that the immediate reward for application $d$ is a non-increasing function of $w_{j, d}$. Finally, we define $w_j = (x^1_j, \dots, x^M_j)$ as the vector describing the number of flows of class $j$ active across all servers.

\paragraph{Policy.} The admission policy $\pi$ is stationary, i.e., a probability distribution over the state-action space set $\pi: \S \rightarrow \A$.    
In the unconstrained setting, the objective function to maximize for the admission control problem is the expected discounted reward $G_t:= \sum_{t=0}^{\infty} {\gamma^{k} R_{t+1+k}}$ starting from initial state $s_0$. We define the value function
\begin{align*}
v_\pi(s) %= \expvalDist{\pi}{G_t| S(0)=s_0} =\expvalDist{\pi}{\sum\limits_{t=0}^{\infty} \gamma^{t} r \rp{S(t),\pi(S(t))}| S(0)=s_0} % \label{eq:avcost2}
=\expvalDist{\pi}{\sum\limits_{t=0}^{\infty} \gamma^{t} r \rp{S(t),\pi(S(t))}| S(0)=s} % \label{eq:avcost2}
\end{align*}

For every stationary deterministic policy, the resulting Markov chain is regular, meaning it has no transient states and a single recurrent non-cyclic class \cite{feinberg1994}. Next, we introduce the CMDP formulation to account for the physical constraints of the system considered, particularly the constraint on access capacity.

%%%%%%%%%%%%%%%%%%%%%%%%%%%%%%%%%%%%%%%%%%%%%%%%%%%%%%%%%%%%%%%%%%%%%%
%%%%%%%%%%%%%%%%%%%%%%%%%%%%%%%%%%%%%%%%%%%%

\section{The CMDP model}\label{sec:cmdp}

%%%%%%%%%%%%%%%%%%%%%%%%%%%%%%%%%%%%%%%%%%%%
%%%%%%%%%%%%%%%%%%%%%%%%%%%%%%%%%%%%%%%%%%%%%%%%%%%%%%%%%%%%%%%%%%%%%%

In CMDP theory \cite{CMDP}, the discounted reward is taken w.r.t. the initial state distribution $\beta:\S \rightarrow \Delta$ :
\begin{equation}\label{eq:avcost1}
	J_\pi(\beta)= \E{s\sim \beta}{v_\pi(s)}
\end{equation}
The access network to server $i$ has capacity $ \theta^i$. Thus, the aggregated long-term throughput demanded by the admitted flows should not exceed such value. We define $c^i: S\times A\rightarrow \R$ the instantaneous cost related to the access  bandwidth constraint: 
\begin{equation}
c^i(s,a) = a \cdot \widehat{c}^i(y^i) 
\label{eq:immediate cost}
\end{equation}
where $\widehat{c}^i$ is an increasing function of $y^i$.\\
The vector $K_\pi(\beta)=(K_\pi^1(\beta),\ldots,K_\pi^M(\beta))$ represents the discounted cumulative constraint, where 
\begin{equation}\label{eq:constraint}
	K_\pi^i(\beta)= \expvalDist{\pi, s\sim \beta}{\sum\limits_{t=0}^{\infty} \gamma^{t} c^i \bigg( X^{i}(t),\pi(S(t)) \bigg) | S(0)=s}  
\end{equation}
For a fixed access capacity vector $\theta=( \theta^1,\ldots, \theta^M)$, and a feasible initial state distribution $\beta$, we seek an optimal policy solving the edge flow admission control (\ref{MVAC}) problem
\begin{align}
	\mathop{\mbox{maximize:}}_{\pi \in \Pi}   \; & J_\pi(\beta)  \label{MVAC}\tag{EFAC}       \\
	\mbox{subj. to:} \;                          & K_\pi(\beta) \leq  \theta \label{eq:constr}
\end{align}
We denote as $J^*(\beta)$ the corresponding optimal value. 

% In the discussion that follows, let define a Bernoulli simple policy \cite{MaCDC1986} as a policy which is deterministic apart for one state. 
The following structural result will be the basis of the SRL algorithm presented in the next section.
\begin{theorem}[]\label{thm:monotone} If the \ref{MVAC} problem is feasible, then \\
    i. There exists an optimal stationary policy $\pi$ which is randomized in at most $M$ states;\\
	ii. Such policy is a deterministic stationary policy if the constraint is not active; \\
	iii. When at least one constraint is active, within the optimal stationary policy outlined in i., each state where the optimal policy is randomized corresponds to a distinct destination server.
\end{theorem}
\begin{proof}
The proof is provided in \extver{appendix:proof structure optimal policy}
\end{proof}
From the computational standpoint, an optimal solution of \ref{MVAC} can be determined by solving a suitable dual linear program CMDP \cite{AltSchw1990}, which depends explicitly on the initial distribution $\beta$.  

The learning approach utilized in the following section is grounded in the lagrangian formulation, which simplifies problem \ref{MVAC} to a non-constrained inf-sup problem \cite{CMDP}. 
\begin{equation}
\hskip-2mm \inf_{\lambda \geq 0}\!\sup_\pi L(\lambda, \pi)= \inf_{\lambda \geq 0} \!\sup_\pi \left[ J(\pi, \beta) - \lambda \left( K(\pi, \beta)- \theta\right) \right] 
\label{eq:equivalent unconstrained problem}
\end{equation}
\begin{comment}
The penalized state-value function writes 
\begin{equation}
v_\pi^{\lambda}(s) = \expvalDist{\pi}{\sum_{t=0}^\infty \gamma^t r^{\lambda}(s,a) \mid s_0 = s }% \notag \\
%= \expvalDist{\pi}{\sum_{t=0}^\infty \gamma^t \rp{r(s, a) - \lambda c(s, a)} \mid s_0 = s }
\label{eq:penalized value function}
\end{equation}
where, fixed  multiplier $\lambda$, $r^{\lambda}(s,a) = r(s, a) - \lambda c(s, a)$ is the penalized reward. The penalized Q-function $Q^{\lambda}$ is defined similarly. 
\end{comment}
Hence, the penalized Q-function for a given policy becomes
\begin{equation*}
Q_{\pi}^{\lambda}(s,a) = \expvalDist{\pi}{\sum_{t=0}^\infty \gamma^t r^{\lambda}(S(t), A(t)) \mid S(0) = s,A(0)=a}% \notag \\
%= \expvalDist{\pi}{\sum_{t=0}^\infty \gamma^t \rp{r(S(t), A(t)) - \lambda c(s, a)} \mid s_0 = s }
\end{equation*}
where, for a fixed multiplier $\lambda$, $r^{\lambda}(s,a) = r(s, a) - \lambda c(s, a)$ is the penalized reward. The penalized value function writes $v_\pi^{\lambda}(s)=Q_{\pi}^{\lambda}(s,\pi(s))$. 

{\noindent \em Remark.} {\em We note that \eqref{eq:constraint} can be considered also in the form of an average constraint \cite{tessler2018reward,CPO}. To this respect, it is worth observing that our solution %DRCPO 
works also for the average reward form of \ref{MVAC}. However, for the sake of comparison with state of the art methods, the discounted form is the most popular formulation in safe reinforcement learning \cite{tessler2018reward}.}

%%%%%%%%%%%%%%%%%%%%%%%%%%%%%%%%%%%%%%%%%%%%%%%%%%%%%%%%%%%%%%%%%%%%%%%%%%%%%%
%%%%%%%%%%%%%%%%%%%%%%%%%%%%%%%%%%%%%%%%%%%%%%%%%%%%

%\vspace{-5mm}
\section{Learning the optimal admission policy}
\label{sec: learning optimal admission policy}
%%%%%%%%%%%%%%%%%%%%%%%%%%%%%%%%%%%%%%%%%%%%%%%%%%%%
%%%%%%%%%%%%%%%%%%%%%%%%%%%%%%%%%%%%%%%%%%%%%%%%%%%%%%%%%%%%%%%%%%%%%%%%%%%%%%
\begin{algorithm}[t]
	\caption{Decomposed Reward Constrained Policy Optimization}
	\label{alg: three timescale algorithm description}
	\begin{algorithmic}[1]
		\State Initialize $\lambda(0)$, initial policy $\pi(0)$
		\For{$k = 0, 1, \dots$}
			\State Initialize $S(0) = \rp{ x^i(0), j(0), i(0) }\sim \beta$
			\For{$t = 0, \dots, T$}
				\State $A(t) \sim  \pi(k) (S(t))$  
				\State $S(t+1), R(t+1) \sim  \texttt{actor}(S(t), A(t))$
				\State Critic update for each component according to \cite{Watkins_Dayan_1992}
                \State Actor update: $\epsilon$-greedy policy
			\EndFor
			\State $\lambda$ update according to \eqref{eq:lagrange multiplier update rule} and  \eqref{eq:lagrangian update}
		\EndFor
	\end{algorithmic}
\end{algorithm}
In situations where transition probabilities \eqref{eq:transprob} are unknown, we can resort to RL algorithms to determine an optimal policy for \ref{MVAC}. We design a model-free safe reinforcement learning (SRL) algorithm to account for both the instantaneous reward \eqref{eq:immediate reward} and cost \eqref{eq:immediate cost}. This section begins with a brief introduction to reward decomposition in reinforcement learning, followed by an explanation of the motivation for the chosen type of safe reinforcement learning (Lagrangian relaxation). The next paragraph describes the introduced algorithm, which is built on these concepts. The section concludes with a brief discussion of the algorithm's convergence properties and a remark on its adaptation to the multi-constraint setting studied.

\paragraph{Reward and cost decomposition.} 
For the system at hand, the full state space $\S$ has cardinality $\Omega(M^{\psi+2})$ where 
$\psi = \max \{\psi^i\}$. A direct tabular RL approach is not viable, as typical in resources allocation problems \cite{MaoHotNets2016}. %The proposed scheme simplifies the actor-critic component by taking advantage of the structure of the underlying MDP. 
Reward decomposition has been introduced in \cite{russell2003q} to decompose a RL agent into multiple sub-agents, where their collective valuations determine the global action. 
In previous works, the method has been applied in the conventional unconstrained setting only, see \cite{suttonHordeScalableRealtime2009,vanseijenHybridRewardArchitecture2017,JuozapaitisExplRL}.
By breaking down the reward function into components within the original setting, the policy improvement step is simplified by considering a separable state-action value function \cite{JuozapaitisExplRL}. To best of the author's knowledge, the algorithm proposed is the first one to consider the idea of reward and cost decomposition in the context of safe reinforcement learning. 

\paragraph{Lagrangian relaxation methods for safe reinforcement learning.}
Our method is rooted in the template SRL actor-critic algorithm for the single constraint case proposed by Borkar in 2005 \cite{borkar2005}. This algorithm prescribes a primal-dual learning procedure to solve the CMDP linear program \cite{CMDP} using a three-timescale framework. The two fast timescales learn the optimal policy using an actor-critic approach for a fixed Lagrange multiplier $\lambda$. %\eqref{eq:unconstrained problem}
The optimal value of $\lambda$ is determined via gradient ascent performed at the slowest timescale.

Compared to other methods like interior point methods \cite{IPO} or trust-region methods \cite{CPO}, this approach considers an estimate of the penalized Q-function, along with the usual Q-learning update rule, which more naturally handles the decomposition operation of both the reward and cost functions.

\paragraph{Decomposed Reward Constrained Policy Optimization.}
%%%%%%%%%%%%%%%%%%%%%%%%%%%%%
This paragraph describes the RL algorithm derived from \cref{alg: three timescale algorithm description} by incorporating the decomposed actor-critic component and the Lagrange multiplier update. % \eqref{eq:multiplier update rule}. 
This algorithm is referred to as Decomposed Reward Constrained Policy Optimization (DRCPO).
%%%%%%%%%%%%%%%%%%%%%%%%%%%%%

In our scenario, a natural option is to identify a component for each pair $(i,d)$ representing a destination edge server $i$ and an application $d$ installed therein. Moreover, fixed component $(i,d)$, we can observe how  a reduced representation of the state, $\tilde{s}=(w_{k, d},y^i, k, m)$, is sufficient to compute the immediate penalized reward for each component.  This observation can be useful in reducing the amount of estimates to compute for each component, as it aggregates several different states of the system. For each of the $\sum_i d^i$ components, we will consider a reduced state space of cardinality $\Omega(M \psi^2)$.

In general, given an arrival of class $k$ routed to server $m$, if the flow is admitted, the reward is non-zero only for the components $(m, d)$, where the application $d$ interested in information flows of class $k$ ($\chi_k(d) = 1$). This is given by
\begin{equation}
    r^{i, \lambda}_d(s, a):=r_d(w_{k, d}, a) - \lambda^i c^i (y^i, a)
    \label{eq:immediate reward component}
\end{equation}
with $r_d (\cdot, a)$ and $c^i(\cdot, a)$ having the properties described in \eqref{eq:immediate reward} and \eqref{eq:immediate cost}, respectively. For all other components the penalized reward is null. \\
Finally, we can define, for each component, the corresponding Q-function, following the usual definition:
$$Q^{i, \lambda}_d(s, a) = \E{\pi}{\sum_{t=0}^\infty \gamma^t r^{i, \lambda}_d (S(t), A(t))|S(0)=s,A(0)=a }$$
In the resulting actor-critic scheme, the critic will feature the aggregated Q-function, with updates being performed for all components $Q_d^i$ at each step, using the traditional learning rule of Q-learning \cite{Watkins_Dayan_1992}.

Regarding the action selection, we observe that, due to the reduced size of the action space, the actor can utilise a simple $\epsilon$-greedy exploration strategy \cite{szepesvariAlgorithmsReinforcementLearning2010}. In doing so, convergence to the optimal solution still occurs, but in the set of $\epsilon$-greedy policies: the resulting policy may be sub-optimal, in sight of \cref{thm:monotone}. Conversely, this approach greatly simplifies the  exploration process by considering just deterministic policies, while significantly reduces the policy search space. As seen in the numerical experiments in \cref{sec: numerical results}, for large values of $M$, the loss in performance becomes negligible. Also, alternative, policy gradient methods which can handle more complex action spaces are possible, but are left as future work.

Finally, the Lagrange multiplier update \cite{borkar2005} is the gradient descent step 
\begin{equation}
    \lambda_{k+1} = \lambda_k - \eta_t \nabla_\lambda L(\lambda, \pi_{\tau})
    \label{eq:lagrange multiplier update rule}
\end{equation}
for suitable values of the learning rate $\eta_t$ (line 10), where  
\begin{equation}
	\nabla_\lambda L(\lambda, \pi) = - \rp{ \expvalDist{s}{K(\pi, \beta)} - \alpha }
	\label{eq:lagrangian update}
\end{equation}
\cref{alg: three timescale algorithm description} outlines the structure of the proposed scheme for the episodic form\footnote{The average reward formulation can be derived as described in \cite{suttonRL} }.

\paragraph{Convergence of DRCPO.}
In this paragraph, for the sake of readability, we will omit the symbol $\lambda$ used to denote the Lagrange multiplier. In \cite{JuozapaitisExplRL} the following results regarding the components of the Q-function obtained following the decomposition approach have been proved:
\begin{proposition} Denote $Q_d^i(s, a)(t)$ the update of the $(i, d)$-th component after $t$ learning update. Under the usual conditions for the convergence of Q-learning \cite{Watkins_Dayan_1992}, $Q_d^i(s, a)(t)$ converges almost surely to the optimal component $Q_d^{i,*}(s, a)$, for every component $(d, i)$ and for every pair $s, a$. Moreover, it holds 
    \begin{equation*}
Q(s, a) (t)= \sum_{i=1}^M \sum_{d \in \D^i} Q^{i}_d \rp{s , a} (t)
\end{equation*}
converges a.s. to the optimal Q-function $Q^*(s, a)$ so that 
$$
Q^*(s, a) = \sum_{i=1}^M \sum_{d \in \D^i} Q^{i, *}_d \rp{s , a}
$$
for every pair $(s, a)$.
\label{prop:convergence component}
\end{proposition}
The convergence of DRCPO to the optimal solution is guaranteed, as stated in the following  
\begin{proposition} \label{prop:convergence}
Under standard assumptions on learning rate of stochastic approximation, DRCPO converges to an optimal solution of \ref{MVAC} w.p.1.
\end{proposition}
\begin{proof}  
The proof is provided in \extver{appendix:proof covergence DRCPO}.
% In particular, the proof given in \cite{borkar2005} for the discounted formulation is seen to hold for the discounted one adopted in this work. 
\end{proof}

Actually, the template described in \cref{alg: three timescale algorithm description} provide some flexibility in the implementation of DRCPO. For scenarios where the system involves large values of $\psi^i$, for instance, the critic component can be replaced by a neural network to estimate the value function component-wuse. Of course, at the cost of losing the guarantees of convergence to an optimal policy. 

{\noindent \em Remark.} {\em Remarkably, the literature on SRL does not provide efficient methods to solve \ref{MVAC} under multiple constraints \cite{liu2021policy}. However, \cref{thm:monotone} shows that, while $\M$ has coupled rewards, constraints are actually independent. This, in combination with the reward decomposition described in the next section, let us perform the parallel of multiple single-constraint learning updates. It's worth noting that this reduction permits to optimize the lagrangian vector component-wise in a single timescale.}
%%%%%%%%%%%%%%%%%%%%%%%%%%%%%%%%%%%%%%%%%%%%%%%%%%%%%%%%%%%%%%%%%%%%%%%%%%%%%%%%%%%%%%%%%%%%%%%%%%%%%%%%%%%%%%%%%%%%%
%%%%%%%%%%%%%%%%%%%%%%%%%%%%%%%%%%%%%%%%%%%%%%%%%%%%%%%%%%%%
\section{\secondProblem}
\label{sec: optimizing routing policy}
Up to this point, we have solved \ref{MVAC} while assuming a given static routing control $\{u_j^i\}$. We now seek to optimize the routing control for the sake of \secondProblem. The objective is hence to maximize the reward of the system w.r.t. to joint admission control and routing: % The objective function to maximise is then 
\begin{equation}
J_{u} (\beta) = \expvalDist{s \sim \beta}{ v_u(s)}
    \label{eq:objective function routing}
\end{equation}
 where $v_u(s)$ is the value function of the \secondProblem policy, defined as
\begin{equation}
v_u(s) := \expvalDist{\pi}{\sum_{t=0}^\infty \gamma^t r\rp{(X(t), J(t), u(J(t)) ), A(t)} \mid S(0) = s}
    \label{eq:value function routing policy}
\end{equation}
and $u(J(t))$ represents the server towards which the incoming flow of class $J(t)$ arriving at time $t$ has been routed to.

However, the analysis of the full Markov system, i.e., the SMDP where the action space encompasses both routing and admission appears extremely challenging, because the actions taken at each state are mutually dependent. 
Its analysis goes beyond the scope of the current work.

The optimization algorithm we propose alternates between two steps: the first one computes the new admission control policy given a fixed \secondProblem policy, according to the results of \cref{sec: learning optimal admission policy}. The second step optimizes the \secondProblem policy for a given admission policy with the following update step 
\begin{equation}\label{eq:stocapprox}
  u_j^i(t+1) =  \Pi\left [ u_j^i(t) + \alpha(t)  \hat{g}_j^i(t) \right ] 
\end{equation}
where ${\alpha}$ is a standard step-size sequence and $\Pi[a]=\max(0,\min(a,1))$ is a projection into $[0,1]$.
The gradient of the total return, denoted as $\hat{g}$, is approximated in the symmetric unbiased form, according to the SPSA algorithm \cite{fu1997optimization} 
\begin{equation*}
    \hskip-2.5mm \hat g_j^i(t) = \frac{ \hat{R}_{\tilde{\pi}} \rp{u_j^i(t) + c(t)\Delta_j^i(t) } - \hat{R}_{\tilde{\pi}} \rp{u_j^i(t) + c(t)\Delta_j^i(t)}}{2c(t) \Delta_j^i(t)}
\end{equation*}
where $c$ is the term of a standard stepsize sequence. $\{ \Delta \}$ is a vector part of a sequence of perturbations of i.i.d. components $\{ \Delta_j^i, i = 1, \dots, M, j = 1, \dots, M \}$ with zero mean and where $\expval{\lvert (\Delta_j^i(t))^{-2} \rvert}$ is uniformly bounded. Since we cannot ensure appropriate conditions on the objective function, namely unimodality or convexity, sequence \eqref{eq:stocapprox} is guaranteed to converge w.p.1 to a local maximum \cite{fu1997optimization}, thanks to some weaker regularity properties proved in \extver{appendix:convergence stochastic approximation}. 

The iterative procedure can continue until the convergence condition with respect to routing probabilities $\{u_j^i\}$ is attained . This procedure does not necessarily converge to the optimal value of the objective function \eqref{eq:objective function routing}, however, in the numerical section we have compared its performances to some heuristic methods and have observed its superiority. Details about the heuristic policies and the pseudocode of the adaptive method can be found in \extver{appendix: heuristic policies}. 

\section{Numerical results}
\label{sec: numerical results}

%%%%%%%%%%%%%%%%%%%%%%%%%%%%%%%%%%%%%%%%%%%%%%%%%%%%%%%%%%%%%
%%%%%%%%%%%%%%%%%%%%%%%%%%%%%%%%%%%%%%%%%%%%%%%%%%%%%%%%%%%%%%%%%%%%%%%%%%%%%%%%%%%%%%%%%%%%%%%%%%%%%%%%%%%%%%%%%%%%%%%%%%

Numerical experiments are divided into three main groups. 

The first reported one compares the performance of our learning algorithm against the state-of-the-art general-purpose algorithm, namely RCPO \cite{tessler2018reward}. RCPO follows the same template outlined in \cref{alg: three timescale algorithm description}: it uses two Neural Networks (NNs) to approximate both the value function (critic) and the policy (actor). In Deep Reinforcement Learning (DRL), NNs act as an interpolator, greatly reducing the number of represented states \cite{MaoHotNets2016}. RCPO has been implemented to incorporate the full system state $S(t) = \rp{X(t), J(t), I(t)}$ as input for the neural networks. In the second experiment, we investigate how the reward varies as the number of applications installed per server increases.
The last experiment performs a comparative analysis of the \secondProblem policies proposed in \cref{sec: optimizing routing policy}.  
The system parameters for all experiments were randomly sampled from predefined sets: they are provided in \extver{appendix:details numerical experiments}. Each column in the table represents the sets considered to generate the corresponding data. % All experiments were conducted on a computing system equipped with multiple Intel Xeon Processor E5640 processors, each having a maximum turbo frequency of 2.93 GHz. The CPU has 4 cores and 8 threads.  \TBD
%As there is no access to the optimal policy, for each run, convergence is considered achieved when three consecutive values adhere to constraint \eqref{eq:constraint} with a deviation of less than 5\% and the reward remain within 5\% of the average of the $5$ best results attained in the run. 

\paragraph{Learning the optimal admission policy.}The results depicted in \cref{fig:comparison learning methods} illustrate a comparison of the performance among various admission control algorithms. The load balancing policy is uniform. Without loss of generality, the system hosts only one application per flow class, potentially installed across multiple servers. In this first test, the scenario features $10$ applications per server, ensuring precisely one application per flow class on each server.
The experiments encompassed a total of $20,000$ episodes, with policy evaluations conducted every $100$ episodes. Due to the heterogeneity of the tested system environments, the reward of each sample is normalized w.r.t. to the unconstrained optimal reward, while the cost is normalized w.r.t. the value of the constraint. The performance of DRCPO is compared with RCPO and also with a naive baseline policy. The baseline policy admits flows only when the server's total occupation is below a specific fixed threshold. The threshold value has been optimized to ensure feasibility while maximizing the reward.

The findings from \cref{fig:comparison learning methods} demonstrate that in the conducted experiments, DRCPO consistently outperforms RCPO in terms of the reward, while also demonstrating better compliance to access bandwidth constraints. Specifically, the results reveal that DRCPO achieves convergence to the optimal solution in fewer than $10^4$ episodes on average, whereas RCPO requires about twice the number of episodes.\\
Furthermore, when comparing the two SRL approaches with the previously described naive baseline, it becomes apparent that the solution provided by DRCPO yields a reward approximately $15\%$ higher, while the performance of RCPO are comparable to those of the naive policy based on the server's total occupation.
\begin{comment}
In Tab.~\ref{tab:experiment 1 time} we compared the average execution speed, measured both as absolute time (seconds) and as relative speedup compared to RCPO. RCPO execution time is, on average, over 2 times slower than other algorithm, primarily due to the time spent for NN updates. It is also noteworthy to compare the RAM utilization, with the NN-based solution demonstrating higher resource consumption.    
\end{comment}
\begin{comment}
   {\noindent \em Remark.} {\em  As indicated by these findings, DRCPO attained a reward 15\% higher than the baseline (RCPO) across quite diverse environments. It achieved convergence to the optimal solution in approximately $50$\% of the learning episodes and exhibited modest memory consumption, since it does not require the use of neural networks. The important performance gain of DRCPO have to be ascribed to the fact that the proposed algorithm is duly tailored to the specific problem structure, thus easily outperforming general purpose SRL solutions.}  
\end{comment}

\begin{figure}[t]
    \centering
    \includegraphics[width = \linewidth]{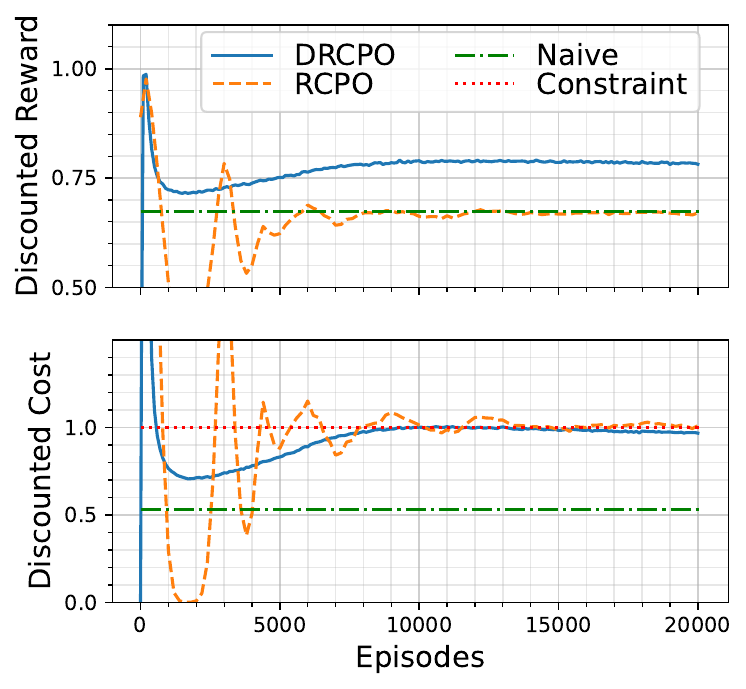}
    \caption{Learning dynamics for the a) discounted reward and b) discounted cost function.}%The plot shows how DRCPO outperforms both the baseline method, RCPO, and the naive threshold policy: in particular, it achieves higher values of the discounted reward and, when compared with RCPO, requires fewer learning episodes to converge to the optimal solution.}
    \label{fig:comparison learning methods}
    % \vspace{-3mm}
\end{figure}
   
\begin{comment}
    \renewcommand{\arraystretch}{1.3}
    \centering
    \begin{tabular}{|c|c|c|}
        \hline
          
         \rowcolor{gray!30} {\bf Learning algorithm} &  {\bf Time (s)} & {\bf Memory consumption (GB)} \\
         \hline
         DRCPO (Sec \ref{subsec: value function decomposition}) & \TBD & \TBD  \\
         RCPO (\cite{tessler2018reward}) & \TBD & \\
         \hline
    \end{tabular}
    \caption{Comparison of cumulative average time for learning and evaluation phases (central column), and peak RAM consumption during learning phase (right column) for different algorithms in the experiments illustrated in Figure \ref{fig:experiment 1}.}
    \label{tab:experiment 1 time}
\renewcommand{\arraystretch}{1.3}
    \centering
    \begin{tabular}{|c|r|wc{3.5em}|wc{3.5em}|wc{3.5em}|}
        \hline
          \rowcolor{gray!30}& & \multicolumn{3}{c|}{ {\bf Speed-up factor w.r.t. RCPO}}  \\
         \cline{3-5}
         \rowcolor{gray!30} \multirow{-2}{*}{{\bf Learning algorithm}} & \multirow{-2}{*}{ {\bf Time (s)}} & {\bf Min} & {\bf Average} & {\bf Max} \\
         \hline
         DRCPO (Sec \ref{subsec: value function decomposition}) & ??? & ?? & ?? & ?? \\
         RCPO (\cite{tessler2018reward}) & ??? & & &\\
         \hline
    \end{tabular}
    \caption{Cumulative average time (in seconds) for learning and evaluation phases of different algorithms in the experiments illustrated in Figure \ref{fig:experiment 1} and speed-up factor of the custom algorithms w.r.t. the baseline general purpose method.}
    \label{tab:experiment 1 time}
\end{comment}
\begin{figure}[t]
    \centering
    \includegraphics[width = \linewidth]{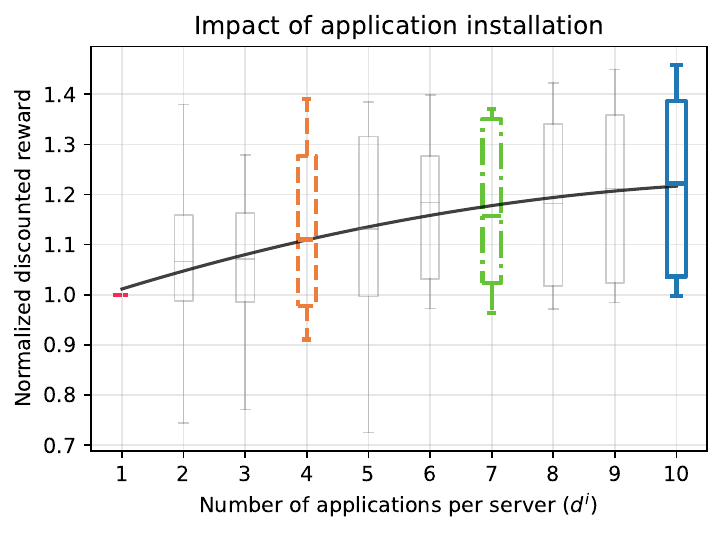}
    \caption{Optimal reward distribution at the increase of the number of applications per server. }
    %\caption{Distribution of the best feasible reward achieved across 20 experiments as the number of applications per server increases. The highest value obtained in each experiment and for each number of applications per server is recorded, ensuring that at least 50\% of constraints are respected and any additional violations remain below 5\%. This criterion accommodates the tendency of discounted costs to closely approach constraints while occasionally surpassing them. The black regression line depicts increasing median rewards with diminishing growth rates as server applications increase.}
    \label{fig:impact application installation}
    %\vspace{-3mm}
\end{figure}
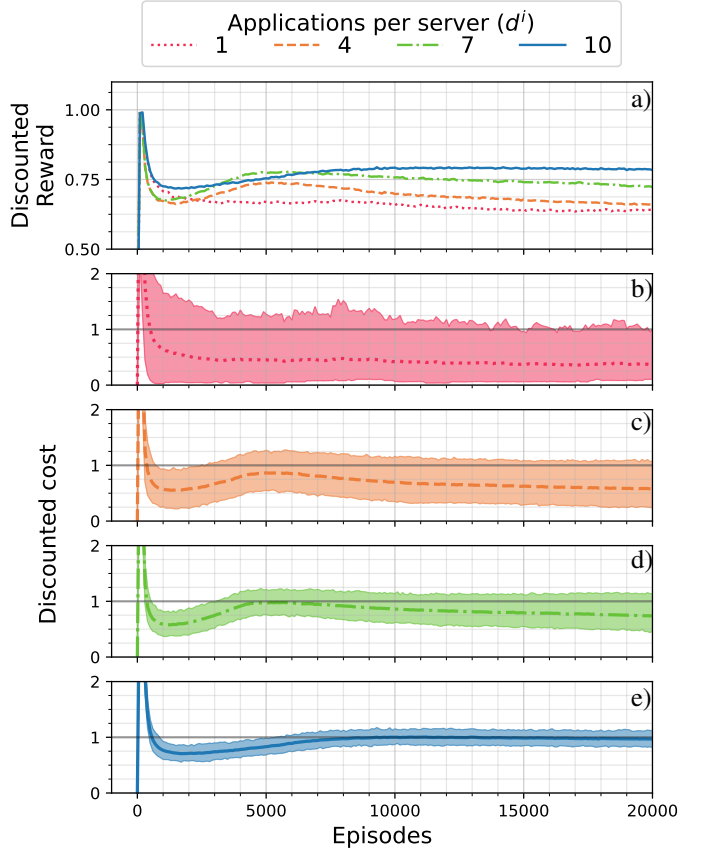
\begin{figure}[t]
    \centering
    \includegraphics[width = \linewidth] {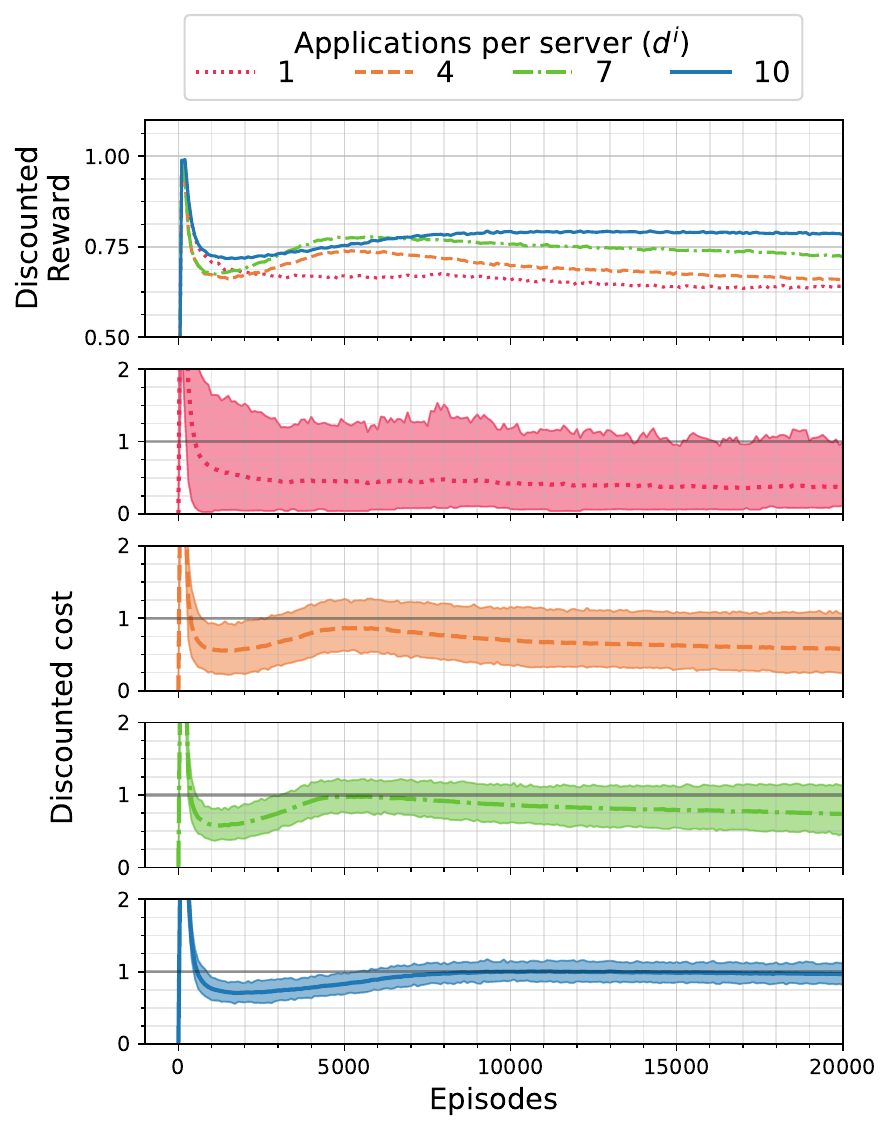}
        %\put(-22,264){a)\put(-80,-2){$d^i \in \{ 1, 4, 7, 10 \}$}}\put(-22,198){b)\put(-50,-2){$d^i=1$}}\put(-22,151){c)\put(-50,-2){$d^i=4$}}\put(-22,102){d)\put(-50,-2){$d^i=7$}}\put(-22,55){e)\put(-52,-2){$d^i=10$}}
\put(-22,284){a)} \put(-22,212){b)}\put(-22,162){c)}\put(-22,110){d)}\put(-22,58){e)}        \caption{The joint reward and Learning dynamics for various values of $d^i$; in (a), (b), (c),  (d), and (e) the discounted cost dynamics. Dashed line: median. Upper and lower borders of the shaded regions: server with highest and lowest associated cost, respectively.}
    %\caption{The top image illustrates the evolution of the discounted reward across 20 experiments with varying numbers of applications installed on each server. Meanwhile, the four bottom images depict the progression of discounted costs. Each plot corresponds to a specific number of applications per server, showcasing the discounted cost of the worst and best servers, as well as the average of all costs. The area between the extreme curves is then colored to provide a better understanding of the difference between them. The data points and corresponding lines represent averages derived from 20 distinct experiments. \\
    %The plots illustrating the evolution of the cost functions not only show how increasing the number of applications per server ensures that each server's cost function closely aligns with its respective constraint but also reflect on the values of the discounted reward. This correlation arises from the fact that a more effective utilization of available resources leads to higher discounted rewards.}
    \label{fig:comparison cost function}
    %\vspace{-3mm}
\end{figure}
%%%%%%%%%%%%%%%%%%%%%%%%%%%%%%%%%%%%%%%%%%%%%%%%%%%%%%%%%%%%%%%%%%%%%%%%%%%%%%%%%%%%%%%%%%
\paragraph{Impact of Application Installation.} We conduct $20$ distinct experiments and analyze a system consisting of $10$ servers, flows classes, and applications.
We observe the trend followed by the optimal discounted reward as the number $k$ of servers hosting each application increases. It's worth noting that in each experiment, exactly $k$ applications are installed on each server, and each application is installed on exactly $k$ servers.

Although the experiments span all possible values of $d^i$, \cref{fig:impact application installation,,fig:comparison cost function} specifically highlight the results for $d^i \in \{1, 4, 7, 10\}$. To facilitate comparison, the same color is maintained for each value throughout the figures in the section. Notably, the linestyle chosen for the case with $d^i = 10$ in \cref{fig:impact application installation,,fig:comparison cost function} matches the line corresponding to DRCPO in \cref{fig:comparison learning methods}, as they represent the same data. Again, in order to compare different results, the values appearing in \cref{fig:impact application installation} are normalized with respect to the optimal discounted reward which is obtained, in each experiment, assuming applications are installed on just one server. The plot displays median data values, and a box plot represents the data distribution. Furthermore, a quadratic regression line illustrates the overall  trend of the median data. The highest value obtained in each experiment and for each number of applications per server is recorded, ensuring that at least $50$\% of constraints are respected and any additional violations remain below $5$\%. This criterion accommodates the tendency of discounted costs to closely approach constraints while occasionally surpassing them. % The black regression line depicts the increasing median rewards with diminishing growth rates as the number of applications per server increase.

From \cref{fig:impact application installation}, we first observe that installing edge applications on all servers appears to be the configuration with best performance across most of the examined experiments. In particular, with just one application per server ($d^i = 1, \forall \ i$), certain servers end up receiving a disproportionately large volume of flows. This is the case when they host applications interested in flows classes with very high arrival rates (or very low departure rates). As a result, feasibility constraints on the access bandwidth require them to admit only a small fraction of flows. On the other hand, as the number of applications per server increases, load balancing attenuates the presence of such hot-spots. The increase of the long term reward eventually levels off: it reaches a plateau towards the highest possible value, %as showed by the black 
as it shows the black regression line which depicts the trend of the median rewards. %In particular, in 
In these experiments the mean increase in reward is around $20\%$ passing from $d^i = 1$ to $d^i = 10$. % applications per server.

Finally, \cref{fig:comparison cost function} provides further insight into the learning dynamics for reward and cost, respectively, for different number of applications per server. The top plot reports on the learning dynamics for the discounted reward, averaged and normalized across all the experiments: it is clear that for higher values of $d^i$ the discounted reward is higher and the convergence to an optimal solution is faster. The subsequent plots in \cref{fig:comparison cost function}b/c/d/e represent the learning dynamics of the cost function as the number of applications per server increases, namely for $d^i \in \{ 1, 4, 7, 10 \}$, respectively. In these plots, the upper (lower) boundary of the colored area denotes the dynamics of the cost of the server with highest (lowest) associated cost, which may change as the number of episodes increase. The line in the middle denotes the average cost across all the servers. 
It is apparent that the difference between the highest cost and the lowest is substantially higher in the case with $d^i = 1$, for the same reason previously described. On the contrary, this difference decreases at the increase of $d^i$. In the extreme case $d^i = 10$ all servers consistently maintain costs proximal to their respective constraints. A higher number of applications per server apparently grants more efficient utilization of available resources, and consequently it increases the discounted rewards across most of the sample data. \\
% The poor performances observed with lower values of $\#\mathcal{D}^i$ is due to the fact that the optimal policy of the system is stochastic in a single state for every server, as demonstrated in Theorem \ref{thm:monotone}. For a lower number of applications, the state with a randomized action is visited more frequently, thereby resulting in worsened performances.  
Another reason behind the poor performances in the case with lower values of $d^i$ is that, as previously mentioned, the implementation of DRCPO presented here exclusively adopts deterministic policies for practical reasons, while, as indicated in \cref{thm:monotone}, the optimal policy is stochastic in one state per destination server. The deterministic nature of the sought policy has a particularly adverse effect on the performance of DRCPO, especially for lower values of $d^i$, as observed in \cref{fig:comparison cost function}. This is likely because the state with the optimal stochastic policy is more frequently visited in these cases.\\
Developing an SRL algorithm that incorporates stochastic policies to address this specific issue, which is notably problematic only in scenarios with low values of $d^i$, would have been more challenging, slower, and ultimately of limited practical utility for more realistic scenarios involving multiple applications per server. The search for an efficient method to derive a policy that is stochastic in a single state is left for future work.

%%%%%%%%%%%%%%%%%%%%%%%%%%%%%%%%%%%%%%%%%%%%%%%%%%%%%%%%%%%%%%%%%%%%%%%%%%%%%%%%%%%%%%%%%%
\paragraph{Comparing different \secondProblem policies.} The last set of numerical experiments of \cref{fig:experiment 3} compares different \secondProblem policies in a scenario where each server may have different parameters. %Additionally, we assume that each application is installed on a single server, although a server can host multiple applications interested in the same class.
% The tested environments feature servers with parameters sampled as outlined in the rightmost column of \cref{tab:parameters experiments} in \extver{appendix:details numerical experiments}.
These experiments examine an increasing number of servers, i.e., $M \in \{3, 5, 7, 10 \}$. Once the arrival rates ${\zeta_j}$ from class $j$ are defined, the flow arrival rate per class to a designated server is determined based on the routing policy.

For ease of comparison, the values on the y-axis of \cref{fig:experiment 3} are normalized relative to the naive \secondProblem, which routes flows to all servers uniformly at random, regardless of whether they host applications interested in such flows. Consequently, all displayed values exceed unity.

\cref{fig:changing load balancing a} illustrates the normalized reward per server concerning the increasing average ratio between servers' capacity ${\psi^i}$ and access capacity ${\theta^i}$. Notably, for lower values of the ratio $\psi / \theta$, the occupation-based \secondProblem policy exhibits the best performance. In this scenario, the reward shows a non-increasing trend concerning a server's occupation: with low $\psi / \theta$ values, the optimal admission policy tends to admit flows more frequently. Consequently, the occupation-based \secondProblem policy prioritizes routing flows to underloaded servers, attaining the highest cumulative reward. However, as the ratio per server $\psi / \theta$ increases, the performance gap in favour of the adaptive \secondProblem widens. Finally, both the uniform and origin-based policies yield comparable results across the experiments.

\cref{fig:changing load balancing b} depicts the normalized reward as the number of servers increases. Once again, the adaptive \secondProblem consistently outperforms all other methods. Furthermore, as the value of $M$ increases, the advantage over the naive \secondProblem widens.

In these experiments, applications installed on each server are interested, on average, in half of the possible flow classes. Consequently, with increasing $M$, the naive uniform load balancing becomes increasingly inefficient, leading to a significant fraction of flows being discarded. However, it's worth noting that while the adaptive \secondProblem demonstrates clearly superior performance, this comes at the cost of a larger number of policy evaluation steps. 

\begin{comment}[t!]
     \centering
     \begin{subfigure}[c]{0.40\columnwidth}
         \centering
         \includegraphics[width=\textwidth]{./FIG/arrival_departure.pdf}
         \caption{System events}\label{fig:arrival_and_departure}
     \end{subfigure}
     \hfill
     \begin{subfigure}[c]{0.58\columnwidth}
         \centering
         \includegraphics[width=\textwidth]{./FIG/state.pdf}
         \caption{System state}\label{fig:state}
     \end{subfigure}
        \caption{(a) {\em Camera arrival and departure:} a camera arrives in area, it is bound to a server, then departs; (b) {\em System state:} the boxed index indicates the areas' designated server each camera is bound to. Using Table \ref{tab:notation} notation: $M = 4$; $Q^1 = \{ A, B \}; Q^2 = \{ A, C \}; Q^3 = \{ B \}; Q^4 = \{ A \}$; $X^1 = (2, 1, 0, 0); X^2 = (0, 2, 0, 0), X^3 = (1, 0, 2, 0); X^4 = (0, 0, 0, 2)$; $j = 1, i= 1$; $Y^1 = 3, Y^2 = 2, Y^3 = 3, Y^4 = 2$.}
        \label{fig:system}
        \vspace{-4mm}
\end{comment}

\begin{figure}
\begin{center}
    \begin{subfigure}[c]{\linewidth}
    \hspace{1cm}
        \includegraphics[width = .8\linewidth]{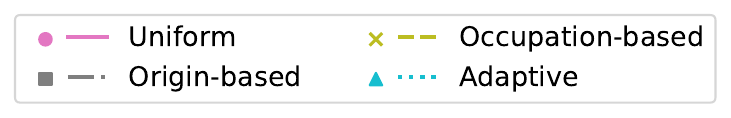}
    \end{subfigure}
    \begin{subfigure}[c]{\linewidth}
    \hspace{.4cm}
        \includegraphics[width = .9\linewidth]{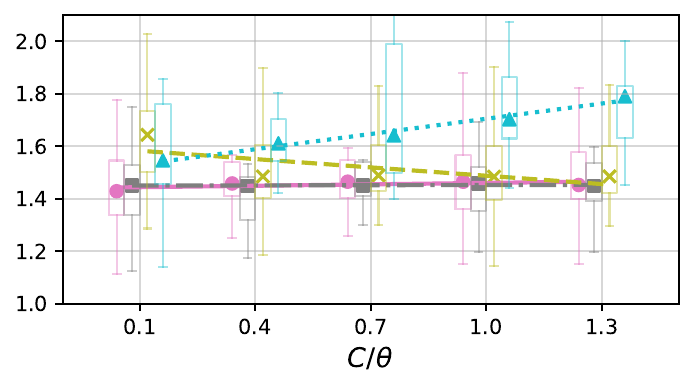}
        %\subcaption{increasing ratio $\psi / \theta$ per server}
        \label{fig:changing load balancing a}
    \end{subfigure}
    \begin{subfigure}[c]{\linewidth}
    \hspace{.4cm}
        \includegraphics[width = .9\linewidth]{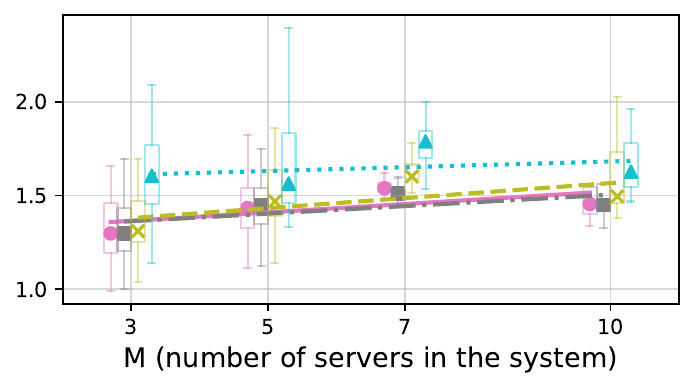}
        %\subcaption{increasing number of areas $M$.}
        \label{fig:changing load balancing b}
    \end{subfigure}
    \caption{Performance of different \secondProblem policies; values on the y-axis normalized w.r.t. naive uniform load balancing: (a)increasing ratio $\psi / \theta$ per server and (b)increasing number of areas $M$.} %Throughout all the experiments, the results obtained using the load balancing policy provided by the adaptive algorithm consistently outperformed the other heuristic policies proposed.}
    \label{fig:experiment 3}
        \end{center}
\end{figure}
\section{Conclusions}\label{sec:conclusions}
%%%%%%%%%%%%%%%%%%%%%
%%%%%%%%%%%%%%%%%%%%%%%%%%%%%%%%%%%%%%%%%%%%%%%%%%%%%%%%%%%%%%%%%%%
Pushed by the surge of edge analytics, the integration of flows from diverse classes poses a significant challenge to existing edge computing architectures. In response, we have introduced a decomposed, constrained Markovian framework for the decentralized admission control of varied information flows. The objective is maximizing the utility of edge applications while accounting for constraints on access network bandwidth and compute capacity. Within this framework, we adopt safe reinforcement learning as the solution concept to derive an optimal policy, even in presence of unknown and highly heterogeneous system parameters. Leveraging the structure of the underlying Markovian model, our proposed solution outperforms state-of-the-art approximated deep reinforcement learning approaches, reducing significantly the number of required learning episodes for convergence. Moreover, our novel reward decomposition method, DRCPO, attains an optimal admission policy.

This work marks an initial step in the field of admission control for edge analytics, opening several directions for future investigation. A  particularly challenging one involves developing a comprehensive Markovian model for joint admission control and routing. Furthermore, addressing network constraints beyond edge access capacity requires considering also the core network topology and the requirements of application modules deployed beyond edge servers. Additionally, one could introduce specific application performance metrics into the model. This would permit to obtain specialized admission policies for flow analytic tasks such as video analytics or anomaly detection. In this regard, model extensions to incorporate per-flow information content are left as part of future work.

\bibliography{biblio}              % Add the filename of your bibliography

\begin{thebibliography}{10}

\bibitem{CPO}
J.~Achiam et~al.
\newblock Constrained policy optimization.
\newblock In {\em Proc. of International Conference on Machine Learning
  (ICML)}, 2017.

\bibitem{afifi2021reinforcement}
H.~Afifi, F.~J. Sauer, and H.~Karl.
\newblock Reinforcement learning for admission control in wireless virtual
  network embedding.
\newblock In {\em Proc. of IEEE ANTS}, 2021.

\bibitem{CMDP}
E.~Altman.
\newblock {\em Constrained Markov Decision Processes}.
\newblock Chapman and Hall, 1999.

\bibitem{AltSchw1990}
E.~Altman and A.~Schwartz.
\newblock Adaptive control of constrained {Markov} chains.
\newblock {\em IEEE Transactions on Automatic Control}, 36(4), 1991.

\bibitem{ananthanarayanan2017real}
G.~Ananthanarayanan, P.~Bahl, and P.~B. et~al.
\newblock Real-time video analytics: The killer app for edge computing.
\newblock {\em IEEE Computer}, 50(10), 2017.

\bibitem{blanc1992optimal}
J.~Blanc, P.~R. de~Waal, P.~Nain, et~al.
\newblock Optimal control of admission to a multiserver queue with two arrival
  streams.
\newblock {\em IEEE Trans. on Automatic Control}, 37(6), 1992.

\bibitem{borgioli2023real}
N.~Borgioli, L.~{Thi Xuan Phan}, F.~Aromolo, et~al.
\newblock Real-time packet-based intrusion detection on edge devices.
\newblock In {\em Proc. of Cyber-Physical Systems and Internet of Things Week}.
  2023.

\bibitem{borkar2005}
V.~Borkar.
\newblock An actor-critic algorithm for constrained {Markov} decision
  processes.
\newblock {\em Elsevier Systems \& Control Letters}, 54(3), 2005.

\bibitem{bronzino2021traffic}
F.~Bronzino, P.~Schmitt, S.~Ayoubi, et~al.
\newblock Traffic refinery: Cost-aware data representation for machine learning
  on network traffic.
\newblock {\em Proc. of the ACM on Measurement and Analysis of Computing
  Systems}, 5(3):1--24, 2021.

\bibitem{brown1998optimizing}
T.~Brown, H.~Tong, and S.~Singh.
\newblock Optimizing admission control while ensuring quality of service in
  multimedia networks via reinforcement learning.
\newblock {\em Proc. of NIPS}, 11, 1998.

\bibitem{chen2018optimized}
X.~Chen, H.~Zhang, C.~Wu, S.~Mao, et~al.
\newblock Optimized computation offloading performance in virtual edge
  computing systems via deep reinforcement learning.
\newblock {\em IEEE Internet of Things Journal}, 6(3), 2018.

\bibitem{Cherkasova2001}
L.~Cherkasova and P.~Phaal.
\newblock Session-based admission control: A mechanism for peak load management
  of commercial web sites.
\newblock {\em IEEE Trans. Comput.}, 51(6), 2002.

\bibitem{choudhary2016video}
A.~Choudhary and S.~Chaudhury.
\newblock Video analytics revisited.
\newblock {\em IET Computer Vision}, 10(4), 2016.

\bibitem{feinberg2006}
X.~Fan-Orzechowski and E.~A. Feinberg.
\newblock Optimality of randomized trunk reservation for a problem with a
  single constraint.
\newblock {\em Advances in Applied Probability}, 38(1), 2006.

\bibitem{feinberg1994}
E.~A. Feinberg and M.~I. Reiman.
\newblock Optimality of randomized trunk reservation.
\newblock {\em Probability in the Engineering and Informational Sciences},
  8(4), 1994.

\bibitem{fu1997optimization}
M.~Fu and S.~Hill.
\newblock Optimization of discrete event systems via simultaneous perturbation
  stochastic approximation.
\newblock {\em IIE Transactions}, 29(3), 1997.

\bibitem{hu2023edge}
M.~Hu, Z.~Luo, A.~Pasdar, et~al.
\newblock Edge-based video analytics: A survey.
\newblock {\em arXiv preprint arXiv:2303.14329}, 2023.

\bibitem{HuangDynAC2022}
J.~Huang, B.~Lv, Y.~Wu, et~al.
\newblock Dynamic admission control and resource allocation for mobile edge
  computing enabled small cell network.
\newblock {\em IEEE Transactions on Vehicular Technology}, 71(2), 2022.

\bibitem{hung2018videoedge}
C.~Hung, G.~Ananthanarayanan, and P.~e.~a. Bodik.
\newblock Videoedge: Processing camera streams using hierarchical clusters.
\newblock In {\em Proc. of the IEEE/ACM Symposium on Edge Computing (SEC)},
  2018.

\bibitem{jiang2018chameleon}
J.~Jiang, G.~Ananthanarayanan, P.~Bodik, et~al.
\newblock Chameleon: scalable adaptation of video analytics.
\newblock In {\em Proc. of ACM SICOMM}, 2018.

\bibitem{JuozapaitisExplRL}
Z.~Juozapaitis, A.~Koul, A.~Fern, et~al.
\newblock Explainable reinforcement learning via reward decomposition.
\newblock In {\em Proc. of IJCAI}, 2019.

\bibitem{Khan2019ECSmartCities}
L.~U. Khan, I.~Yaqoob, N.~H. Tran, et~al.
\newblock Edge-computing-enabled smart cities: A comprehensive survey.
\newblock {\em IEEE Internet of Things Journal}, 7:10200--10232, 2019.

\bibitem{Konstanteli2014}
K.~Konstanteli, T.~Cucinotta, K.~Psychas, et~al.
\newblock Elastic admission control for federated cloud services.
\newblock {\em IEEE Transactions on Cloud Computing}, 2(3), 2014.

\bibitem{lilith2005using}
N.~Lilith and K.~Dogancay.
\newblock Using reinforcement learning for call admission control in cellular
  environments featuring self-similar traffic.
\newblock In {\em Proc. of IEEE TENCON}, 2005.

\bibitem{lippman1974}
S.~A. Lippman.
\newblock Applying a new device in the optimization of exponential queuing
  systems.
\newblock {\em Operations Research}, 23(4), 1975.

\bibitem{IPO}
Y.~Liu, J.~Ding, and X.~Liu.
\newblock Ipo: Interior-point policy optimization under constraints.
\newblock In {\em Proc. of AAAI}, 2019.

\bibitem{liu2021policy}
Y.~Liu, A.~Halev, and X.~Liu.
\newblock Policy learning with constraints in model-free reinforcement
  learning: A survey.
\newblock In {\em Proc. of IJCAI}, 2021.

\bibitem{LuoEdgeSurvey2021}
Q.~Luo, S.~Hu, C.~Li, et~al.
\newblock Resource scheduling in edge computing: A survey.
\newblock {\em IEEE Communications Surveys \& Tutorials}, 23(4), 2021.

\bibitem{MaoHotNets2016}
H.~Mao, M.~Alizadeh, I.~Menache, et~al.
\newblock Resource management with deep reinforcement learning.
\newblock In {\em Proc. of ACM HotNets}, New York, NY, USA, 2016.

\bibitem{IGA}
A.~Massaro, F.~De~Pellegrini, and L.~Maggi.
\newblock Optimal trunk-reservation by policy learning.
\newblock In {\em Proc. of IEEE INFOCOM}, 2019.

\bibitem{miller1969queueing}
B.~L. Miller.
\newblock A queueing reward system with several customer classes.
\newblock {\em Management science}, 16(3), 1969.

\bibitem{millnert2018achieving}
V.~Millnert, J.~Eker, and E.~Bini.
\newblock Achieving predictable and low end-to-end latency for a network of
  smart services.
\newblock In {\em Proc. of IEEE GLOBECOM}, 2018.

\bibitem{pakha2018reinventing}
C.~Pakha, A.~Chowdhery, and J.~Jiang.
\newblock Reinventing video streaming for distributed vision analytics.
\newblock In {\em Proc. of USENIX HotCloud}, 2018.

\bibitem{paternain2019constrained}
S.~Paternain, L.~Chamon, M.~Calvo-Fullana, and A.~Ribeiro.
\newblock Constrained reinforcement learning has zero duality gap.
\newblock {\em Advances in Neural Information Processing Systems}, 32, 2019.

\bibitem{puterman2014markov}
M.~L. Puterman.
\newblock {\em Markov decision processes: discrete stochastic dynamic
  programming}.
\newblock John Wiley \& Sons, 2014.

\bibitem{raeis2020reinforcement}
M.~Raeis, A.~Tizghadam, and A.~Leon-Garcia.
\newblock Reinforcement learning-based admission control in delay-sensitive
  service systems.
\newblock In {\em Proc. of IEEE GLOBECOM}, 2020.

\bibitem{russell2003q}
S.~J. Russell and A.~Zimdars.
\newblock Q-decomposition for reinforcement learning agents.
\newblock In {\em Proc. of ICML}, 2003.

\bibitem{senouci2004call}
S.~Senouci, A.~Beylot, and G.~Pujolle.
\newblock Call admission control in cellular networks: a reinforcement learning
  solution.
\newblock {\em International journal of network management}, 14(2), 2004.

\bibitem{seufert2024marina}
M.~Seufert, K.~Dietz, N.~Wehner, et~al.
\newblock Marina: Realizing ml-driven real-time network traffic monitoring at
  terabit scale.
\newblock {\em IEEE Transactions on Network and Service Management}, 2024.

\bibitem{Sue2011}
C.~Sue, Y.~Hsu, and P.~Ho.
\newblock Dynamic preemption call admission control scheme based on {Markov}
  decision process in traffic groomed optical networks.
\newblock {\em Journal of Optical Communications and Networking}, 3(4), 2011.

\bibitem{Sajal_OSDI2023}
M.~Sultan, L.~Marshall, B.~Li, et~al.
\newblock Kerveros: Efficient and scalable cloud admission control.
\newblock In {\em Proc. of USENIX OSDI}, 2023.

\bibitem{suttonRL}
R.~S. Sutton and A.~G. Barto.
\newblock {\em Reinforcement Learning: An Introduction}.
\newblock A Bradford Book, Cambridge, MA, USA, 2018.

\bibitem{suttonHordeScalableRealtime2009}
R.~S. Sutton, J.~Modayil, M.~Delp, et~al.
\newblock Horde: A scalable real-time architecture for learning knowledge from
  unsupervised sensorimotor interaction.
\newblock In {\em Proc. of AAMAS}, 2011.

\bibitem{szepesvariAlgorithmsReinforcementLearning2010}
C.~Szepesvari.
\newblock {\em Algorithms for reinforcement learning}.
\newblock Number~9 in Synthesis lectures on artificial intelligence and machine
  learning. Morgan \& Claypool, 2010.

\bibitem{tessler2018reward}
C.~Tessler, D.~Mankowitz, and S.~Mannor.
\newblock Reward constrained policy optimization.
\newblock In {\em Proc. of ICLR}, 2019.

\bibitem{vanseijenHybridRewardArchitecture2017}
H.~van Seijen, M.~Fatemi, J.~Romoff, et~al.
\newblock Hybrid reward architecture for reinforcement learning.
\newblock In {\em Proc. of NIPS}, 2017.

\bibitem{wan2022retina}
G.~Wan, F.~Gong, T.~Barbette, et~al.
\newblock Retina: analyzing 100gbe traffic on commodity hardware.
\newblock In {\em Proc. of the ACM SIGCOMM}, pages 530--544, 2022.

\bibitem{wang2020joint}
C.~Wang, S.~Zhang, Y.~Chen, et~al.
\newblock Joint configuration adaptation and bandwidth allocation for
  edge-based real-time video analytics.
\newblock In {\em Proc. of IEEE INFOCOM}, 2020.

\bibitem{wang2018bandwidth}
J.~Wang, Z.~Feng, Z.~Chen, et~al.
\newblock Bandwidth-efficient live video analytics for drones via edge
  computing.
\newblock In {\em Proc of IEEE/ACM SEC}, 2018.

\bibitem{wang2017hast}
W.~Wang, Y.~Sheng, J.~Wang, et~al.
\newblock Hast-ids: Learning hierarchical spatial-temporal features using deep
  neural networks to improve intrusion detection.
\newblock {\em IEEE Access}, 6:1792--1806, 2017.

\bibitem{Watkins_Dayan_1992}
C.~J. Watkins and P.~Dayan.
\newblock Q-learning.
\newblock {\em Machine Learning}, 8(3/4), 1992.

\bibitem{zhang2019hetero}
W.~Zhang, S.~Li, L.~Liu, et~al.
\newblock Hetero-edge: Orchestration of real-time vision applications on
  heterogeneous edge clouds.
\newblock In {\em Proc. of IEEE INFOCOM}, 2019.

\bibitem{zhang2020decomposable}
Y.~Zhang, J.~Liu, C.~Wang, et~al.
\newblock Decomposable intelligence on cloud-edge {IoT} framework for live
  video analytics.
\newblock {\em IEEE Internet of Things Journal}, 7(9), 2020.

\end{thebibliography}
\bibliographystyle{abbrv}

%\end{document}
%%%%%%%%%%%%%%%%%%%%%%%%%%%%%%%%%%%%%%%%%%%%%%%%%%%%%%%%%%%%%%%%%%%%%%%%%%%
% \begin{comment}
\clearpage

\appendix

\section{Appendix}
\subsection{Use cases}
\label{appendix:use cases}
%%%%%%%%%%%%%%%%%%%%%%%%%%%%%%%%%%%%%%%%
%%%%%%%%%%%%%%%%%%%%%%%%%%%%%%%%%%%%%%%%%%%%%%%%%%%%%%%%%%%%%%%%%%%
In this section, we present the general characteristics of data intensive AI-based applications deployed at the edge, as well as two prominent use cases. Modern edge applications can be commonly characterized by five features: applications process flows generated by a large number of sources of different nature (\ding{202}); these flows can enter or leave the architecture over time due to various events (\ding{203}); the edge infrastructure deploys a set of applications (\ding{204}) to process the flows on edge servers which are equipped with a given amount of resources (e.g., compute and memory) (\ding{205}); finally, the distributed nature of both sources and edge servers imposes the implementation of a control plane mapping flows to compute infrastructure (\ding{206}).

Two practical use cases of AI-based applications are video analytics and anomaly detection in network traffic. We highlight their characteristics, as well as how admission control can impact their performance. 

% must be selected by a controller in order to 
% balance the resource occupancy and the satisfaction of the applications running on the edge servers. In this case the admission control is based on the importance of each flow given by the 
% information that the video is bringing to applications interested to the geographical area where the video comes from. The second use-case is the anomaly detection performed at the edge of the network. Indeed, with the spread of high-speed connections and the use of ML models to detect anomalies in network traffic, controller must select, in real-time, the most important flows that can contribute to an accurate detection of traffic anomalies. The rest of the section describes the two use-cases in detail.  

\paragraph{Video Analytics.} 
Video analytics are a core application of edge
clouds~\cite{ananthanarayanan2017real}. In video analytics, applications typically consist of cascades of functions performing operations (e.g., decoding, background extraction, and object detection) on
incoming video streams (\ding{202})~\cite{choudhary2016video, hung2018videoedge}. Recent literature on video analytics focuses on the problem of orchestrating the deployed applications on edge infrastructure~\cite{jiang2018chameleon,wang2020joint,zhang2020decomposable}. The orchestration
controls the placement of application functions across edge servers (\ding{204}), abiding to infrastructure constraints (\ding{205}), and determines how video traffic data is routed between
different application functions (\ding{206}) across the edge
infrastructure. Yet, as mobile video cameras become more and more widespread, the number of video sources is in constant grow. As mobile cameras join and depart the system (\ding{203}), the number of sources present in the system exceeds the total capacity of the compute infrastructure, two solutions become available: reduce the computational complexity of the applications deployed or select subsets of flows to process. While the first solution is commonly used in the literature~\cite{hung2018videoedge}, this comes at the expense of the applications' accuracy. Conversely, by selecting which sources to process, the system has the potential of maximizing accuracy and minimizing overhead by removing redundant processing (e.g., cameras that have overlapping field of view). Yet, this requires the implementation of admission control algorithms, the core of this paper. Figure~\ref{fig:system} summarizes the described video analytics application where cameras can arrive and depart (Figure~\ref{fig:arrival_and_departure}) and are mapped to existing infrastructure (Figure~\ref{fig:state}).
 
\paragraph{Anomaly Detection.} Network anomaly detection defines the methods that target the identification of anomalous behaviors in network traffic caused by unusual activities that may indicate a security breach. Modern anomaly detection techniques~\cite{bronzino2021traffic,wan2022retina,seufert2024marina,borgioli2023real} commonly involve the use of machine learning (ML) models to monitor traffic flows (\ding{202}) and map them to critical events. For this purpose, practical solutions manipulate raw network flows, transforming packets into representations that are amenable for input to the ML models~\cite{bronzino2021traffic,wan2022retina,seufert2024marina}. Such transformations go from aggregate statistics (e.g., calculation of flow sizes) to more complex operations (e.g., gray-scale image representations~\cite{wang2017hast}). However, as traffic representations increase in complexity, monitoring systems are forced to filter out set of flows that can be relevant for detection~\cite{bronzino2021traffic}. To this end, admission control techniques become crucial for ensuring an effective real-time monitoring by selecting flows of interest from the network's traffic.  In this scenario, a controller becomes in charge of evaluating the relevancy of each new incoming flow and decide whether to admit to the processing pipeline (\ding{203}). In case a flow is admitted, compute and memory resources are allocated (\ding{205}) to compute statistics used by the ML model for the treatment of the flow. Such monitoring tasks are conventionally performed within the access infrastructure, i.e., the edge (\ding{204}),  to avoid overheads and bottlenecks generated transferring copies of the traffic to a centralized location at the core of the network. Furthermore, the controller can select for each flow the monitoring node in the edge infrastructure(\ding{206}) and reduce  network overhead and detection response time~\cite{borgioli2023real}.

\subsection{Related works}
\label{appendix:related}
%%%%%%%%%%%%%%%%%%%%%%%%%%%%%%%%%%%%%%%%%%
%%%%%%%%%%%%%%%%%%%%%%%%%%%%%%%%%%%%%%%%%%%%%%%%%%%%%%%%%%%%%%%%%%%%%%%%%%%%%%%%%%%%%%%%%%%%%%%%%%%%%%%%%%%%%

Our proposed solution for multi server edge-systems extends early models of \textit{trunk-reservation} for single-server loss systems \cite{miller1969queueing,blanc1992optimal}. There, for a single constraint, the optimal solution is a stairway-type threshold policy~\cite{feinberg1994,feinberg2006}. Optimal threshold policies for one-dimensional state-dependent rewards have been studied early on for single queues, see, e.g., \cite{lippman1974}. 
In our case, however, the system is multi-server, rewards are multi-dimensional and state-dependent, so that monotone policies are difficult to characterize beyond the single server case addressed, e.g., in \cite{IGA}. 

The works \cite{brown1998optimizing,lilith2005using,senouci2004call,IGA} use RL for the admission control of QoS differentiated traffic classes, possibly coupled to resource allocation mechanisms to address their specific resource requirements \cite{lilith2005using}. More recently, network function chaining for jobs with end-to-end deadlines was addressed in \cite{millnert2018achieving}. The authors of \cite{afifi2021reinforcement} provided heuristic QoS guarantees for wireless nodes by cascading admission control and resource allocation via DRL. In \cite{raeis2020reinforcement} a multi-server queuing system with a non-linear delay constraint resorts to RL coupled with a Lagrangian-type heuristics. In all those works, and, to the best of the authors' knowledge, in the related literature, the multi-server admission control problem under access network constraints has not been addressed so far. 
%Notably, admission decisions are made irrespective of the various classes of data already present in the queue, and the proposed safe reinforcement learning algorithm employs a fixed value for the Lagrange multiplier, in contrast to utilizing an adaptive mechanism as in other approaches.
In the edge computing literature, reward decomposition has been used recently for task offloading towards multiple edge servers in \cite{chen2018optimized}, but in the unconstrained setting. Conversely, this work leverages the CMDP framework to provide a provably optimal SRL solution via reward decomposition. 

\subsection{Proof of Theorem \ref{thm:monotone}}\label{appendix:proof structure optimal policy}
%%%%%%%%%%%%%%%%%%%%%%%%%%%%%%%%%%%%%%%%%%%%%%%%%%%%%%%%%%%%%%%%%%%%%%%%%%%

Before the actual proof, we need some preliminary results. Let us partition the state space $\S=\bigcup_i \S^i$ where $\S^i=\{s \in \S \mid s=(x,j,i),\, \forall x,\, \forall j \}$. That is, $\S^i$ is the set of states whose destination server is $i$. Let us define 
$\rho(s,a)$ the stationary state-action distribution. From a know result in CMDP theory \cite{CMDP}  
\begin{lemma}\label{lem:dualLP}
An optimal stationary state-action distribution $\rho^*$ for \ref{MVAC} solves the following dual linear program
\begin{align}
\mathop{\mbox{maximize:}}_{\rho(s,a)}   \; & \sum_{i=1}^M \sum_{s \in \S^i, a \in \A(s)}r(s,a) \rho(s,a)   \label{DLP}\tag{DLP}                                                                                                       \\
\mbox{subj. to:}                 & \sum_{i=1}^M \sum_{s \in \S^i, a \in \A(s)} \rho(s,a) [ \delta_{s\underline s}- \gamma p(\underline s |s,a)] = \beta(\underline s), \; \underline s\in \S \label{eq:bellman}             \\
	                                     & \sum_{s \in \S^i, a \in \A(s)} c^i(s,a) \rho(s,a)\leq  \theta_i, \; 1\leq i\leq M  \label{eq:constraintApp}                                                        \\
	                                     & \rho(s,a)\geq 0, \quad \forall (s,a)  \in \S \times \A  \label{eq:nonnegApp}
\end{align}
where $\beta$ is the initial distribution and $\delta_{s\underline s}=1$ if $s=\underline s$ and zero otherwise. The corresponding optimal policy writes, for non transient states, as  
\begin{equation}
\pi^*(a|s)=\frac{\rho^*(s,a)}{\sum_{a'\in \A} \rho^*(s,a')}
\end{equation}
\end{lemma}
We now provide the main proof. 
\begin{proof}
\noindent i. The statement follows directly from a general result for finite CMDPs \cite{CMDP}[Thm. 3.8]. \\
\noindent ii. When constraint \eqref{eq:constr} is not active for an optimal policy, an optimal deterministic stationary policy exists \cite{puterman2014markov}. \\
\noindent iii.  We assume at least one active constraint, otherwise we fall into the case described in ii. We start from an optimal pair $\pi^*$ and $\rho^*$ solving the dual LP as described \cref{lem:dualLP}.  Hence, we iterate an exchange argument over the partition $\{\S^i\}$ of the state space in order to obtain a solution which is not worse off and has the required property.  Let us consider $\S^1$ without loss of generality, and define the CMDP $\M_o=(\S_o,\A,P_o)$ with state space $\S_o=\S^1$ and with constraint \eqref{eq:constraintApp} corresponding to $\theta_1  P_1$ where $P_1=\sum_{s \in \S_1, a \in \A(s)} \rho^*(s,a)$. The action set $\A'(s)=\A(s)$ for $s\in \S^1$. Finally, the transition probabilities $p_o$ and the transition rewards  $r_o$ of $\M_o$ are those induced by $\pi^*$ for first-return transitions from $\S_1$ into $\S_1$. Let us consider an optimal solution $\rho_o^*(s,a)$ for the corresponding dual program of $\M_o$: \begin{align}
\mathop{\mbox{maximize:}}_{\rho_o(s,a)}   \; & \sum_{s \in \S^1, a \in \A(s)}r(s,a) \rho_o(s,a)   \label{DLPsingle}\tag{DLPo}                                                                                                       \\
\mbox{subj. to:}                 &  \sum_{s \in \S^1, a \in \A(s)} \rho_o(s,a) [ \delta_{s\underline s}- \gamma p_o(\underline s |s,a)] =  \frac{\beta(\underline s)}{\beta_1}, \; \underline s\in \S^1 \label{eq:bellmanSingle}             \\
	                                     & \sum_{s \in \S^1, a \in \A(s)} c^1(s,a) \rho_o(s,a)\leq P_1\theta_1,   \label{eq:constraintSingle}                                                        \\
	                                     & \rho_o(s,a)\geq 0, \quad \forall (s,a)\in \S_1\times \A   \label{eq:nonnegSingle}
\end{align}
where distribution $\beta_1 = \sum_{s \in \S_1}  \beta( s) $. By applying Thm 3.8 in \cite{CMDP} to \ref{DLPsingle}, we deduce the existence of an optimal policy $\rho_o^*(s,a)$  for $\M_o$, which is randomized in at most one state. 
Now, consider the normalized distribution $\frac{\rho^*(s,a)}{P_1}$ defined on $ \S^1 \times \A$: it respects the constraint  (\ref{eq:constraintSingle}) and solves \ref{DLPsingle}. Hence, $\sum_{s \in \S^1, a \in \A(s)}r_o(s,a) \rho_o^*(s,a) \geq \sum_{s \in \S^1, a \in \A(s)}r_o(s,a) \rho^*(s,a)/P_1$. 

We can now define the following state-action probability distribution for the original MDP $\M$: 
\begin{eqnarray}\label{eq:dist}
\rho_o^{\star}(s,a)= 
\begin{cases}
\rho^*(s,a)&\text{if} \quad s\not \in \S^1\cr
\rho_o^*(s,a) P_1   &\text{if} \quad s\in \S^1
\end{cases}
\end{eqnarray}
Clearly, $\rho_o^{\star}$ is still a solution of  \ref{DLP}. Furthermore it is not worse off $\rho_o^*(s,a)$, which concludes the proof. 
\end{proof}

%%%%%%%%%%%%%%%%%%%%%%%%%%%%%%%%%%%%%%%%%%%%%%%%%%%%%%%%%%%%%%%%
\subsection{Proof of Proposition \ref{prop:convergence}}
\label{appendix:proof covergence DRCPO}
%%%%%%%%%%%%%%%%%%%%%%%%%%%%%%%%%%%%%%%%%%%%%%%%%%%%%%%%%%%%%%%%
% In the general case with $M$ constraints, it is probably possible to consider $M+2$ timescales to prove the convergence, but this is left as future work. We argue that in the case studied, thanks to the properties of the constraints, which are independent to each other, it is possible to converge to the optimal solution considering only $3$ timescales, as in the single constraint case. 
\begin{proof}
(Sketch) The convergence of the learning algorithm to the optimal solution requires a stochastic approximation argument. First, \cite{tessler2018reward}, a sketch of proof is outlined for the convergence of the template $3$-timescale constrained actor-critic learning algorithm in a scenario with discounted reward and a single cost function, without resorting to reward decomposition. However, in the general case with $M$ constraints, this algorithm may not converge to the optimal solution since the studied constrained MDP problem is not convex.\\
Nevertheless, in \cite{paternain2019constrained}, it is shown that, in general, the CMDP problem, here \ref{MVAC}, and its corresponding dual problem, here \eqref{eq:equivalent unconstrained problem}, have no duality gap, provided that Slater's condition is satisfied. In particular, the main result there implies that the problem becomes convex in the dual domain. In the context of CMDPs, meeting the Slater's condition translates to having a policy that satisfies all constraints \cite{CMDP}. %Achieving this isn't always straightforward and can be non-trivial in various cases. 
In our CMDP model, this condition is indeed fulfilled by the feasible policy $\pi \equiv 0$. This establishes the equivalence between the two problems. Hence, \cite{paternain2019constrained} guarantees that an optimal solution exists and can be determined, e.g., alternating value iteration and dual gradient ascent when the kernel is known. %The actual convergence of the learning algorithm to the optimal solution requires a stochastic approximation argument. 
The convergence of the $3$-timescale stochastic approximation iteration is hence proved by applying the ODE method developed in \cite{borkar2005} for the case of an actor-critic with multiple constraints.
\end{proof}
\

%%%%%%%%%%%%%%%%%%%%%%%%%%%%%%%%%%%%%%%%%%%%%%%%%%%%%%%%%%%%%%%%
\subsection{Heuristic load balancing policies}
\label{appendix: heuristic policies}
The heuristic policies considered in our experiments are {\em uniform}, {\em origin-based}, {\em occupation-based} and {\em adaptive} \secondProblem. They are described as follows:

{\em Uniform:} flows are routed uniformly at random towards all available servers. 

{\em Origin-based:} it routes the flows of class $j$ based on the departure rates and the access bandwidth constraints of destination servers, namely using routing probability $ u_j^i = \exp (\theta^i  \mu_j ) /\sum_k \sum_h \exp (\theta^k  \mu_h )$. The rationale is that flows of classes with higher departure rate $\mu_j$ engage lesser servers' resources, and destination server with higher access constraint $\theta^i$ handles more flows per time unit without exceeding its capacity. 

{\em Occupation-based:} the routing probability to destination server $j$ accounts the servers' state and access capacity $$u_j^i(x^1, \dots, x^M, j) = \exp(-x_j^i / \theta_i)\sum_k \exp(- x_j^k / \theta_k)$$In order to define the  routing probabilities $u_j^i$, this routing policy requires an initial admission policy, obtained using the uniform \secondProblem. Its consequent evaluation gives the corresponding ergodic occupancy distribution per server and therefore the routing probabilities. A further evaluation step for admission control policy ensures the respect of admission capacity constraints \eqref{eq:constr}.

\subsection{Additional details on the convergence proof of the adaptive load algorithm}
\label{appendix:convergence stochastic approximation}
%%%%%%%%%%%%%%%%%%%%%%%%%%%%%%%%%%%%%%%%%%%%%%%%%%%%%%%%%%%%%%%%
In order to prove the convergence to a local optimal point of the stochastic approximation method, it is necessary to establish certain regularity properties of the objective function:
\begin{proposition}
\label{prop:convergence stochastic approximation}
%Fixed the admission policies for each server, the objective is differentiable with regard to routing probabilities $\{u_j^i\}$, i.e., $\nabla_u J_u(\beta)$ exists.
Fixed the admission policies for each server, $J_u(\beta)$ is continuously differentiable in the routing probabilities $\{u_j^i\}$. 
\end{proposition}
\begin{proof} 
Let consider admission policies  $\tilde{\pi}$ fixed for each device and denote as $P_u$ the transition probability matrix of the system. Clearly, all the entries of $P_u$ are differentiable w.r.t. $\{u_j^i\}$, as they depend on the transitions of the single servers and in \eqref{eq:transition probabilities} we could write $\p_j^i = u_j^i \lbda_j / \Lbda^i$. To conclude, the objective function can be computed as 
\begin{equation}
J_u(\beta) = \beta \rp{ I + \gamma P_{\tilde{\pi}} + \gamma^2 P_{\tilde{\pi}}^2 + \dots } r_{\tilde{\pi}}
\end{equation} 
where $r_{\tilde{\pi}}$ is the per-state instantaneous reward vector. 
\end{proof}

%%%%%%%%%%%%%%%%%%%%%%%%%%%%%%%%%%%%%%%%%%%%%%%%%%%%%%%%%%%%%%%%
\subsection{Pseudocode for the adaptive load balancing algorithm}

\begin{algorithm}[h]
    \caption{Adaptive Load Balancing}
    \label{alg:alternating routing improvement}
    \begin{algorithmic}[1]
    \Require $\epsilon > 0$, episode length $T$ 
    \State  {\bf Init:} uniform \secondProblem
    \State  {\bf Calculate:} optimal admission policy $\pi^i$, $i=1,\ldots,M$
    \While $\left| u_{j, old}^i - u_{j, new}^i \right| > \epsilon$\\
    $\quad u_{j, old}^i \gets u_{j, new}^i$\\
    $\quad${\bf Improve the \secondProblem policy:} by \eqref{eq:stocapprox}\\
    $\quad${\bf Learn an optimal admission policy:}     method in \cref{sec: learning optimal admission policy}, for fixed episode length $T$
    \EndWhile \\
    \Return $u_{j, new}^i$
    \end{algorithmic}
\end{algorithm}

\subsection{Parameters of the numerical experiments in Section \ref{sec: numerical results}}
\label{appendix:details numerical experiments}
%%%%%%%%%%%%%%%%%%%%%%%%%%%%%%%%%%%%%%%%%%%%%%%%%%%%%%%%%%%%%%%%

\begin{table}[H]
\renewcommand{\arraystretch}{1.3}
    \centering
    \begin{tabular}{|c|c|c|c|}
    \hline
		\rowcolor{gray!30}{\bf Parameter}  & {\bf \cref{fig:comparison learning methods}} & {\bf \cref{fig:impact application installation,,fig:comparison cost function}} & {\bf \cref{fig:experiment 3}} \\
         \hline
         $\gamma$ & \multicolumn{3}{c|}{$[0.95, 1)$} \\
         \hline
         $M$ &  \multicolumn{2}{c|}{$M$}& $\{ 3, 5, 7, 10 \}$  \\
         \hline
         $d^i$ & $10$ & \multicolumn{2}{c|}{$\{1, \dots, M \}$}\\
         \hline
         $\phi_d$ & $10$ & $d^i$ & 1\\
         \hline
         $\psi^i$ & \multicolumn{2}{c|}{$\{ 20, 21, \dots, 30 \}$} & $\{ 8, \dots, 16 \}$\\ 
         \hline
         $\theta^i$ & \multicolumn{2}{c|}{$1/\rp{20(1-\gamma)}$} & $\left[ \psi/10, 3\psi/2 \right]$ \\
         \hline
         $u_j^i$ & \multicolumn{2}{c|}{$ 1/M$} & -\\
         \hline
         $\lbda_j$ & \multicolumn{2}{c|}{$[1, 2)$} & $\left[ 0, 1.5 \right)$ \\
         \hline
         $\mu_j$ &  \multicolumn{2}{c|}{$\left[ 0, 0.5 \right)$} & $\left[ 0, 0.2 \right)$ \\
         \hline
         \multirow{3}{*}{r(s, 1)} & \multicolumn{2}{c|}{\multirow{2}{*}{$a \exp \{ - w_{j, d} \frac{b}{M} \} + c$}} &  i. $ \frac{M}{j+1}$\\  
         & \multicolumn{1}{c}{} & & ii. $\frac{M}{j+1} \exp \{ x_j^i \}$\\
         & \multicolumn{2}{c|}{$a, b \in \left[ 1, 5 \right]$, $c \in \left[ 0, .1 \right]$} & iii. $\frac{M}{j+1}   \exp \{ -x_j^i \frac{j}{2M} \}$\\
         \hline
         $c(s, 1)$ & \multicolumn{2}{c|}{$y^i$}& $y^i/M$\\
         \hline
    \end{tabular}
    \caption{System parameters used in the numerical experiments.} 
\label{tab:parameters experiments}
\end{table}

% \end{comment}

\end{document}